\documentclass[11pt]{article}

\usepackage[T1]{fontenc}
\usepackage[utf8]{inputenc}
\usepackage{lmodern}
\usepackage{graphicx}
\usepackage{booktabs}
\usepackage{natbib}
\usepackage{booktabs}
\usepackage{makecell}
\usepackage{adjustbox}
\usepackage{subcaption}
\usepackage{float}
\usepackage{thmtools}
\usepackage[colorlinks=true,
linkcolor=blue,
citecolor=blue,
urlcolor=blue]{hyperref}

\RequirePackage{amsthm,amsmath,amsfonts,amssymb,enumitem}
\usepackage[a4paper,margin=1in]{geometry}

\theoremstyle{plain}

\newtheorem{theorem}{Theorem}[section]
\newtheorem{lemma}[theorem]{Lemma}
\theoremstyle{definition}
\newtheorem{definition}{Definition}[section]
\newtheorem{example}{Example}[section]

\newtheorem{corollary}[theorem]{Corollary}
\newtheorem{proposition}[theorem]{Proposition}
\newtheorem{remark}[theorem]{Remark}

\title{On Response-Adaptive Targeting Strategies for Multi-Treatment Experiments}
\author{
	\textbf{Redouane Yagouti}\\
	Univ. Lille, Inria,  CNRS, Centrale Lille\\ 
	UMR 9189-CRIStAL, France\\
	\texttt{redouane.yagouti@inria.fr}
	\and
	\textbf{Rémy Degenne}\\
	Univ. Lille, Inria,  CNRS, Centrale Lille\\ 
	UMR 9189-CRIStAL, France\\
	\texttt{remy.degenne@inria.fr}
	\and
	\textbf{Emilie Kaufmann}\\
	Univ. Lille, CNRS, Inria,  Centrale Lille\\
	UMR 9189-CRIStAL, France\\
	\texttt{emilie.kaufmann@univ-lille.fr}
}

\date{\today}

\begin{document}
	
	\maketitle
	
	\begin{abstract}		

Response-adaptive randomization (RAR) in clinical trials aims to improve ethical and statistical efficiency by dynamically allocating patients to treatments based on observed outcomes. While RAR based on a target optimal allocation have been extensively studied for two-arms settings, their extension to multi-treatment experiments ($K \geq 2$) remains theoretically fragmented, with most existing methods focusing on specific algorithms or restricted target allocations.
In this paper, we introduce a unified framework for response-adaptive targeting, the $\alpha$-Rebalancing Targeting Strategies ($\alpha$RTS), which generalize the ERADE two-armed strategy of \cite{hu2009efficient}. We prove that all designs in this family share fundamental asymptotic properties: strong consistency, asymptotic normality of allocation proportions and treatment effect estimators, and asymptotic efficiency. To address sparse target regimes (where some treatments are asymptotically eliminated), we further propose $\alpha$RTS with Forced Exploration, a variant that guarantees infinite sampling for all treatments while preserving the asymptotic guarantees.
Extensive simulations illustrate the finite-sample behavior of $\alpha$RTS variants in a 3-armed context, highlighting in particular the critical role of forced exploration in sparse settings.

%

	\end{abstract}
	
	\section{Introduction}
	Equal randomization is the gold standard in confirmatory (phase III) clinical trials: each treatment, also called arm, is allocated to approximately the same number of patients, after which some statistical test is performed to identify promising treatments. Using a Response Adaptive Randomization (RAR) procedure, i.e., randomizing patients to treatments based on the observed outcomes for previous patients, has a great potential to make clinical trials more effective: fewer patients may be needed for inferring treatment effects or patient benefit could be increased. Still, the actual use of RAR in clinical trials has been debated for years, as explained  by \cite{Villar23Review}.

There exists many RAR procedures. The first is attributed to \citep{Thompson33}, and variants of this Bayesian strategy have been explored in clinical trials \citep{Thall07}. In the 1950s multi-armed bandits models were introduced to model adaptive selection between different populations \citep{Robbins52Freq} . Since then, bandit algorithms have become a research field of their own \citep{BanditBook}. However, many of these adaptive algorithms are actually \emph{not} randomized hence do not qualify as RAR procedures. Moreover, the type of performance measures considered in that literature are not the typical testing metrics common in clinical trials. Another class of RAR procedures are urn-based designs, such the Randomized Play the Winner (RPW) \citep{WeiDurham78RPW} that was actually used in a real trial \citep{Bartlett85ECMO} but with mixed results as the resulting allocation was extremely unbalanced. In the 2000s, a family of RAR procedures based on a so-called \emph{optimal allocation proportion} emerged \citep{Rosenberger01Opt}. Given some assumptions on the arms distributions and some statistical test performed at the end of the trial, these strategies first derive the allocation optimizing some criterion for a fixed power of the test. For example, minimizing the total number of patients needed for a fixed power of a Wald test in the case of two-armed binary outcomes leads to the so-called Neyman allocation, but many other target allocations have been considered under different assumptions, see \citep{pin2024response} for a recent survey. These \emph{target allocations} all have in common that they depend on the unknown parameters of the arms distributions, and thus a crucial component of the resulting strategy is a RAR procedure that, given a target allocation, ensures that the empirical allocation proportions converge to that target.

%
Our paper is about such \emph{adaptive targeting mechanisms}. Many such mechanisms were originally studied for two arms, such as the celebrated Biased Coin Design of \cite{Efron71BCD}, which targets a $(1/2,1/2)$ allocation. The work of \cite{hu2004asymptotic} studies the Doubly adaptive Biased Coin Design for general (smooth) targets allocations and further extend this design to the multi-armed case ($K \geq 2$). As argued in several works (e.g. \cite{HuRosenberger03VariabilityPower}), empirical allocations that are less variable lead to a better power of the statistical test using the adaptively collected data. This motivates the quest for \textit{asymptotic efficiency}, defined as achieving the minimal asymptotic variance for the allocation proportions subject to a given allocation target. For adaptive targeting mechanisms, this notion was formalized by \cite{hu2006asymptotically}, who derived lower bounds for the asymptotic covariance of allocation proportions. 
While DCBD is actually not asymptotically efficient, \cite{hu2009efficient} introduce the Efficient Randomized Adaptive Design (ERADE), in the two arms setting. Their procedure achieves the asymptotic efficiency bound while preserving randomization and allowing for general target allocation functions, including targets derived from optimality criteria. Under mild regularity conditions, ERADE yields strong consistency and asymptotic normality for both allocation proportions and estimators of treatment effects. However, its extension to more than $2$ arms was left as an open question.

\paragraph*{Contributions} In this paper, we propose an extension of ERADE for $K\geq 2$ arms. More precisely, we introduce a family of algorithms keeping the same spirit as ERADE, which we call $\alpha$-Rebalancing Targeting Strategies ($\alpha$RTS). 
Interestingly, the $\alpha$RTS framework recovers the extension of ERADE recently proposed by \cite{alkhnefr2025efficient} in a parallel work, but also allows for different algorithmic choices which we explore in the paper. In particular, we show that a Tracking strategy proposed in a different context in the multi-armed bandit literature \citep{garivier2016optimal} is also an $\alpha$RTS. We propose a unified asymptotic analysis of $\alpha$RTS, establishing consistency and asymptotic normality for both allocation proportions and treatment effect estimators, and we show that the resulting allocation proportions achieve the asymptotic efficiency lower bound.

Prior work only apply to target allocations that are strictly positive, and do not explicitly address sparse target regimes in which some treatments are asymptotically eliminated. Yet, in multi-armed setting, such sparse targets have been considered \citep{tymofyeyev2007implementing}. To overcome this limitation, we propose a variant of the $\alpha$RTS family that incorporates an adaptive forced exploration mechanism into the allocation rule. Despite this modification, we show that our proposed $\alpha$RTS-FE inherits the same asymptotic properties  under the same structural conditions and its asymptotic analysis can be extended to sparse target allocations. In this setting, the adaptive exploration mechanism guarantees infinite sampling of all treatments, ensuring the validity of estimator asymptotics even when certain treatments are asymptotically eliminated.

On the technical side, we emphasize that we propose a unified analysis of the $\alpha$RTS and $\alpha$RTS-FE families, that identifies some sufficient conditions for a RAR design to be asymptotically efficient (Lemma~\ref{lem:generic-conditions-main}). These conditions generalize some conditions that can be  extracted from the original ERADE analysis for $K=2$ arms, and we show that they apply to both type of strategies by introducing some appropriate notion of hitting time $\ell_{n,k}$.

In addition to our theoretical contributions, we complement the asymptotic analysis with a comprehensive simulation study designed to assess the finite-sample behavior of the proposed targeting strategies in a trial with $K=3$ arms with binary outcomes, where an homogeneity tests is performed. These experiments investigate both regular and sparse allocation regimes, comparing different algorithmic instantiations within the \(\alpha RTS\) family and their forced-exploration counterparts. In particular, we evaluate convergence properties of allocation proportions and plug-in estimates, as well as the impact on hypothesis testing performance and statistical power in multi-arm settings. The simulations highlight subtle yet practically important differences between designs that are asymptotically equivalent, and illustrate the critical role of forced exploration in ensuring reliable adaptation in sparse regimes.

\paragraph*{Outline}
The remainder of the paper is organized as follows. Section~\ref{sec:notations} introduces the model, notation, and standing assumptions. Section~\ref{sec:alpha-rts} presents the $\alpha$-Rebalancing Targeting Strategies and discusses particular instantiations. Section~\ref{sec:proof_vanilla} present our main theoretical results and sketch their proof. Section~\ref{sec:alpha-rts-fe} introduces the forced-exploration variant and shows that the asymptotic properties are preserved, including in sparse target allocation regimes. Finally in Section~\ref{sec:simulation} we propose some experiments for different algorithms of the two proposed families. Technical lemmas and detailed proofs are given in the appendices.

	\section{Model, Notation, and Assumptions}\label{sec:notations}
	We consider a sequential clinical trial with $K$ treatments, also called arms. Patients arrive sequentially and are assigned to treatments adaptively. In the following we denote \([K]:= \left\{1, \dots, K\right\}\).

\paragraph*{Arm's distributions}
For each patient $m$, let \(X_m = (X_{m,1}, \ldots, X_{m,K})\) denote the assignment vector, where $X_{m,k}=1$ indicates that the patient \(m\) is assigned to arm \(k \in [K]\) and $X_{m,k}=0$ otherwise, thus we have $\sum_{k=1}^{K} X_{m,k} = 1$.

Consider the sequence of \textit{i.i.d.} random variables \(\left\{\xi_{m,k}, m=1, \dots\right\}\)  where \(\xi_{m,k}\) is the response of the \(m^{th}\) patient on treatment \(k\). For a treatment $k \in [K]$, the parameter of interest in this paper is the expectation of these responses, we denote it as \(\theta_k := \mathbb{E}\left[\xi_{1,k}\right]\). We also define the vector of these parameters, \(\Theta = (\theta_1, \ldots, \theta_K) \in \mathbb{R}^K\).

The total number of patients assigned to each treatment up to time $n$ is denoted by \(N_n = (N_{n,1}, \ldots, N_{n,K})\) where \(N_{n,k} = \sum_{j=1}^n X_{j,k}\).
The normalized vector $N_n/n$ represents the observed allocation proportions at time $n$. Controlling the long-term behavior of these proportions is the central objective of adaptive targeting mechanisms.

\begin{remark}
Prior work often consider a more general setting in which both the observation $\xi_{m,k}$ 
and the parameter of the distribution $\theta_k$ can be multi-dimensional (e.g., $\theta_k$ is the mean and covariance of a Gaussian arm). 
For the sake of clarity, we focus our investigation on the (doubly) uni-dimensional case, for distribution parameterized by their means. This encompasses the example of arms that belong to a one-dimensional exponential family, that has received particular attention in the multi-armed bandit literature \citep{KLUCBJournal,garivier2016optimal}. However,  our analysis and results extend to these more complex settings without substantial modification.
\end{remark}

\paragraph*{Response-Adaptive Randomization (RAR) procedures}

Let $\mathcal{F}_m = \sigma\!\left((X_i,\xi_i)_{1 \le i \le m}\right)$ 
denote the filtration generated by the treatment assignments and observed 
outcomes of the first $m$ patients. A \textit{RAR} procedure is a sequence
$(p_m)_{m\ge1}$ of $\mathcal{F}_m$-measurable random vectors taking values 
in the probability simplex
\[
\Delta_K := \left\{ p \in [0,1]^K : \sum_{k=1}^K p_k = 1 \right\}.
\]
At step $m+1$, the treatment assignment $X_{m+1}$ is drawn conditionally 
on $\mathcal{F}_m$ according to $p_m$, that is,
\[
\mathbb{P}(X_{m+1} = k \mid \mathcal{F}_m) = p_{m,k},
\qquad k = 1,\dots,K.
\]

\paragraph*{Adaptive Targeting Mechanisms} Given some target function $\rho : \mathcal{H} \rightarrow \Delta_K$, i.e. a function such that, for any $z \in \mathcal{H}$,
\[
\rho(z) = (\rho_1(z), \ldots, \rho_K(z))
\]
satisfies $\sum_{k=1}^{K}\rho_{k}(z)=1$, we propose in this work RAR procedures under which the empirical allocation $N_n/n$ converges to the target allocation $v := \rho(\Theta)$. We refer to these procedures as adaptive targeting mechanisms.

Different choices of $\rho$ correspond to different design objectives, such as equal allocation, Neyman allocation for power, or ethically motivated rules \citep{pin2024response}.
Existing adaptive targeting mechanisms include DBCD \citep{hu2004asymptotic} and ERADE \citep{hu2006theory}, which we extend in the next section to $K \geq 2$ arms.

\paragraph*{Notation}
In the paper we will use the notations $o_p(\cdot)$ and $O_p(\cdot)$, which are analogues of the deterministic little-$o$ and big-$O$ notations for convergence in probability.
More formally, for a sequence $\{X_n\}$ of random variables and $\left(a_n\right)_{n \geq 1}$ a sequence of positive real numbers.
We say that $X_n = o_p(a_n)$ if for all $\varepsilon > 0$, \(\lim_{n \to \infty} \mathbb{P}\!\left( \frac{|X_n|}{a_n} > \varepsilon \right) = 0.\)
We say that $X_n = O_p(a_n)$ if for all $\varepsilon > 0$, there exists $M > 0$ such that \(\limsup_{n \to \infty} \mathbb{P}\!\left( \frac{|X_n|}{a_n} > M \right) \leq \varepsilon.\)

	\section{$\alpha$-Rebalancing Targeting Strategy (\(\alpha\)RTS)}\label{sec:alpha-rts}
	We now introduce our main algorithmic contribution: a general family of adaptive targeting mechanisms that extend the ERADE design to $K> 2$ arms. Its first ingredient is an estimator of the unknown parameters and target allocation.

\paragraph*{Sequential estimation} 
Since $\Theta$ is unknown, the parameters $\theta_k$ are estimated online. Given some initial value $\theta_{0,k}$, the sequential estimator for treatment $k$ at stage $m$ is
\begin{equation}\label{eq:dbcd-estimator}
	\hat{\theta}_{m,k} := \frac{\sum_{j=1}^m X_{j,k}\xi_{j,k} + \theta_{0,k}}{N_{m,k}+1}.
\end{equation}
The role of the initial values $\theta_{0,k}$ is to regularize the estimator when few or no patients have yet been assigned to treatment $k$. In practice, the initial values $\theta_{0,k}$ can be chosen either as a common constant across all treatments, i.e., $\theta_{0,k} = \theta_0$, or in a treatment-specific manner.

Letting \(\hat{\Theta}_m := (\hat{\theta}_{m,1}, \ldots, \hat{\theta}_{m,K}),\) our estimate of the target allocation $v=\rho(\Theta)$ is the plug-in estimator $\hat{\rho}_m := \rho(\hat{\Theta}_m).$


\begin{remark}
For our analysis to go through, an alternative estimator of $\theta_k$ can be used, as long as it has the Bahadur-type representation, i.e., it can be written
\begin{equation}\label{eq:bahadur-estimator}
	\hat{\theta}_{m,k} = \frac{1}{N_{m,k}}\sum_{j=1}^{m}X_{j,k}\xi_{j,k} + o(N_{m,k}^{-1/2}),\quad \text{as } m\to\infty.
\end{equation}
\end{remark}

The ERADE design, introduced by \cite{hu2009efficient} for $K=2$ uses a parameter $\alpha \in [0,1)$ and works as follows. It compares the current empirical allocation to arm $1$, $N_{m,1}/m$ to its estimated target $\hat{\rho}_{m,1}$ and whenever $\frac{N_{m,1}}{m}>\hat\rho_{m,1}$ it reduces the probability to select arm $1$ in the next round by setting $p_{m+1,1}=\alpha \hat\rho_{m,1}$.
More precisely, it sets
	\begin{equation*} \label{eq:erade}
		p_{m+1,1} =
		\begin{cases}
			\alpha \hat{\rho}_{m,1}, & \text{if } N_{m,1}/m > \hat{\rho}_m, \\[6pt]
			\hat{\rho}_{m,1}, & \text{if } N_{m,1}/m = \hat{\rho}_m, \\[6pt]
			1 - \alpha(1 - \hat{\rho}_{m,1}), & \text{if } N_{m,1}/m < \hat{\rho}_m,
		\end{cases}
	\end{equation*}
and $p_{m+1,2}=1-p_{m+1,1}$. Hence, the philosophy of ERADE is to use for $p_{m+1}$ a rebalanced version of the allocation vector $\hat{\rho}_m$ that reduces the probability to select an arm that has already be selected more than its current target estimate. This idea can in fact be generalized to more than $2$ arms, leading to our proposed $\alpha$RTS family.

\begin{definition} Let $K\geq 2$ and $\alpha\in[0,1)$. A design is an $\alpha$-Rebalancing Targeting Strategy ($\alpha$RTS) if it a RAR procedure whose allocation probabilities $p_{m+1}=(p_{m+1,1},\dots,p_{m+1,K}) \in \Delta_K$ satisfy the following: for all $m \geq Km_0$,
	\begin{align}
		&\forall k \in [K], \  \left(\frac{N_{m,k}}{m}>\hat\rho_{m,k}\quad \Longrightarrow\quad p_{m+1,k}\le \alpha\,\hat\rho_{m,k}\right) \label{eq:throttle}
	\end{align}
\end{definition}
\begin{remark}
	The parameter $m_0$ (sometimes called burn-in in a clinical trial context) is the number of initial samples from each arm that may be taken before starting adaptive randomization. It is included for practical purposes only and doesn't influence the proofs. We can safely take $m_0 = 0$ if $\theta_0$, the regularization parameter of the estimator, is chosen in an appropriate way.
\end{remark}
We now present a few examples of $\alpha$RTS with $K > 2$ arms. $\alpha$RTS imposes how the probability of over-sampled arms has to decrease (as in \eqref{eq:throttle}), but gives some flexibility in how to correspondingly increase the probability of some under-sampled arms. We first propose a design that increases these probabilities in proportion to their distance to the target.

\begin{example}[Distance Based Allocation]
	Given $\alpha \in [0,1)$, for all $m\geq Km_0$, for all $k \in [K]$,
	\begin{equation*}\label{eq:example}
		p_{m+1,k} \;=\; \alpha \hat{\rho}_{m,k} \;+\; (1-\alpha)\,
		\frac{\delta_{m,k}}{\sum_{i=1}^K \delta_{m,i}},
		\ \text{ where } \
		\delta_{m,k} \;=\; \max\!\left(0,\ \hat{\rho}_{m,k} - \frac{N_{m,k}}{m}\right)\;.
	\end{equation*}
	If for all $j$, $\delta_{m,j}=0$, then for every $j \in \{1,\ldots,K\}$, we choose \(p_{m+1,j} \;=\; \hat{\rho}_{m,j}.\) We use a distance-based design to ensure that arms that deviates the most from their target allocation are corrected more aggressively.
\end{example}

\smallskip

Instead of relying on the distance to the target, an independent work by \citep{alkhnefr2025efficient} proposed an extension of ERADE to $K>2$ that boosts the probability of under-sampled arms by adding the same quantity to all of them.

\begin{example}[ERADE 2025 \citep{alkhnefr2025efficient}]
	Given $\alpha \in [0,1)$\footnote{The paper actually set $\alpha \in (0,1)$ but there is no need for that}, for all $m\geq Km_0$, $k \in [K]$,
	
	\[
	p_{m+1,k} =
	\begin{cases}
		\alpha \hat{\rho}_{m,k} &
		\text{ if } \ \dfrac{N_{m,k}}{m} > \hat{\rho}_{m,k}, \\[6pt]
		\hat{\rho}_{m,k} &
		\text{ if } \ \dfrac{N_{m,k}}{m} = \hat{\rho}_{m,k}, \\[6pt]
		\hat{\rho}_{m,k} + (1-\alpha)\displaystyle\frac{\sum_{j \in S}\hat{\rho}_{m,j}}{|T|} &
		\text{ if } \ \dfrac{N_{m,k}}{m} < \hat{\rho}_{m,k}.
	\end{cases}
	\]
	where $S = \left\{ j \; \middle| \; \dfrac{N_{m,j}}{m} > \rho_{m,j} \right\}$ is the set of over-sampled arms,  $T = \left\{ j \; \middle| \; \dfrac{N_{m,j}}{m} < \rho_{m,j} \right\}$ is the set of under-sampled arms and $|T|$ denotes the cardinal of set $T$.
\end{example}

\smallskip

While \citet{alkhnefr2025efficient} proposed a specific extension of ERADE to $K>2$ arms, our $\alpha$RTS family provides more flexibility in how the sampling budget can be re-allocated to under-sampled arms. This flexibility also reveals an interesting connection with the
	multi-armed bandit literature. Indeed, by concentrating the rebalancing
	effort on a single under-sampled arm, one obtains a family of targeting
	rules that contains the D-Tracking sampling rule of \cite{garivier2016optimal} as a particular case.
	
	\begin{example}[Interpolated D-Tracking]
		Inspired by the D-Tracking rule of \cite{garivier2016optimal}, we define for a parameter $\alpha \geq 0$ the \textit{Interpolated D-tracking} rule by
		\begin{equation}\label{eq:alpha-d-tracking}
			p_{m+1,k}
			=
			\alpha \hat \rho_{m,k}
			+
			(1-\alpha)\mathbf 1_{\{k=k^\star\}},
		\end{equation}
		where \(k^\star
		=
		\arg\max_{a\in[K]}
		\left\{
		m\hat \rho_{m,a}-N_{m,a}
		\right\}
		\). The quantity $m\hat \rho_{m,a}-N_{m,a}$ measures the deficit of arm $a$
		relative to its target allocation. Hence, $k^\star$ is the arm whose
		current number of draws is furthest below its prescribed proportion.
		
		For all $\alpha$, we can easily verify that this design belongs to the $\alpha$RTS family. It clearly satisfies $\sum_{k=1}^{K} p_{m+1,k}=1$. Then, for any arm $k$ that is over-sampled, i.e. for which $N_{m,k}/m > \hat \rho_{m,k}$, we have $m\hat \rho_{m,k}-N_{m,k} < 0$, which implies that $k\neq k^\star$ and \(p_{m+1,k}
		=
		\alpha \hat \rho_{m,k}\).

		In the special case where $\alpha=0$, \eqref{eq:alpha-d-tracking} reduces to \(p_{m+1,k}
		=
		\mathbf 1_{\{k=k^\star\}},\)
		which coincides with the original D-tracking rule of
		\citet{garivier2016optimal} without including its forced exploration component (which will be discussed again in Section~\ref{sec:alpha-rts-fe}).
\end{example}

\section{Asymptotic analysis of $\alpha$RTS} \label{sec:proof_vanilla}

We now present our theoretical guarantees for $\alpha$RTS.

\subsection{Main Asymptotic Results}

In this section, we establish that any design belonging to the $\alpha$RTS family satisfies certain asymptotic properties under mild regularity conditions. More precisely, we show that the allocation proportions converge almost surely to the target allocation, while the parameter estimators and plug-in targets converge at explicit rates. We further characterize the asymptotic fluctuations of both the estimators and the allocation proportions through central limit theorems, and we prove that the resulting allocation proportions achieve the asymptotic efficiency lower bound.

Our analysis relies on two assumptions: one on the response distributions (Condition \hyperlink{CA}{\textbf{A}}) and one on the target allocation (Condition \hyperlink{CB}{\textbf{B}}). The latter explicitly excludes sparse targets, i.e. target allocations for which $\rho_k(z)=0$ for some coordinate $k$, for any value value $z$ that is a candidate estimated value for the parameter $\Theta$ under any allocation strategy.

\begin{description}
	\item[\hypertarget{CA}{\textbf{Condition A}}] For treatment \(k \in [K]\) the corresponding response \(\xi_{1,k}\) satisfies  $
	E|\xi_{1,k}|^{2} < \infty $. \\
	
	\item[\hypertarget{CB}{\textbf{Condition B}}] The target function $\rho$ and the regularization parameter $\theta_0$ of the estimator satisfy the following. Let $I_k = \{\theta_k\} \bigcup V_k$ where $V_k$ is the set of all possible values that can be observed for the estimator $\hat\theta_{n,k}$, for all $n$ and all realizations of the responses $\{\xi_{j,k}\}$. The domain $\mathcal{H} \subseteq \mathbb{R}^K$ of the target function is open and contains $I_1\times \dots \times I_K$. Moreover:
	\begin{itemize}
		\item \(\rho\) is twice differentiable on $\mathcal{H}$,
		\item for all $z \in I_1\times \dots \times I_k$, \ \(\rho(z) \in (0,1)^K\).
	\end{itemize}
\end{description}

In the following we present asymptotic results that are valid for any design belonging to the \(\alpha\)RTS family for \(\alpha \in [0,1)\). We first establish the fundamental convergence properties of the procedure. We show that the allocation proportions converge to the target allocation almost surely, while the parameter estimators and plug-in target estimator converge at explicit rates.

\begin{theorem}[Strong consistency and rates]\label{thm:LLN}
	Under conditions \hyperlink{CA}{\textbf{A}}--\hyperlink{CB}{\textbf{B}}, for \(n \to \infty\),
	\[
	\frac{N_{n,k}}{n}\xrightarrow{\text{a.s.}}v_k,\qquad
	\hat\Theta_n-\Theta = O\!\Big(\sqrt{\frac{\log\log n}{n}}\Big)\text{ a.s.},\qquad
	n(\hat\rho_n-v)=O\!\big(\sqrt{n\log\log n}\big)\text{ a.s.}
	\]
\end{theorem}

%


We now refine these results by characterizing the fluctuations around the limit. The following theorem establishes asymptotic normality for both the estimators and the allocation proportions, as well as their joint distribution.

\begin{theorem}\label{thm:normality-gen-erade}
	We introduce the following notation
	\begin{itemize}
		\item \(G := \nabla (\rho)\big|_{\Theta}\) denotes the gradient of the target allocation $\rho$ with respect to $\Theta$
		\item \(V_k := \mathbb{V}ar[\xi_{1,k}]\)  is the variance for treatment \(k\) response
		\item \(V := diag\left(\frac{1}{v_1}V_1, \dots, \frac{1}{v_K}V_K\right)\) is a diagonal matrix collecting the scaled variances
		\item \(\Omega := 
		\begin{pmatrix}
			GVG^{\top } & GVG^{\top } \\
			GVG^{\top } & GVG^{\top }
		\end{pmatrix}\) is a block matrix built from \(GVG^{\top }\)
	\end{itemize}
	Under conditions \hyperlink{CA}{\textbf{A}}--\hyperlink{CB}{\textbf{B}}, we have as \(n \to \infty\)
	
	\begin{enumerate}[label=(\roman*)]
		\item  For every \(k \in [K]\),
		\begin{align}
			|N_{n,k} - n\hat{\rho}_{n,k}| &= o_P(\sqrt{n})\label{eq:thm-res-1}\\
			|N_{n,k} - n\hat{\rho}_{n,k}| &= O(\sqrt{n \log\log n})\quad a.s.\label{eq:thm-res-2}\\
			N_{n,k}-n v_k&=O(\sqrt{n\log\log n})\quad a.s.\label{eq:thm-res-3}
		\end{align}
		\item The asymptotic normality results:
		\begin{align}\label{eq:thm-res-4}
			\sqrt{n}(\hat{\Theta}_n - \Theta) &\xrightarrow{\mathcal{D}} N(0,V),\\
			\label{eq:thm-res-5}
			\begin{pmatrix}
				\sqrt{n} \left( \frac{N_n}{n} - v \right)\\
				\sqrt{n} (\hat{\rho}_n - v)
			\end{pmatrix}
			&\xrightarrow{\mathcal{D}} \mathcal{N}(0,\Omega).
		\end{align}		
	\end{enumerate}
\end{theorem}

\begin{remark}[asymptotic efficiency] A direct consequence of Theorem~\ref{thm:normality-gen-erade} is that the empirical allocation proportion $N_n/n$ is an asymptotically efficient estimator of the target $v$, at least when the arms distribution belong to an exponential family. Indeed, under this assumption, Theorem 1 of \cite{hu2006asymptotically} provides a lower bound on the asymptotic variance of $N_n/n$ as an estimator of $v$ under any RAR procedure, that is
\[G\,I(\Theta)^{-1}\,G^{\top }\]
where  \(I(\Theta) = \mathrm{diag}\left(v_1 I_1(\theta_1), \dots, v_K I_K(\theta_K)\right)\) and \(I_k(\theta_k)\) is the Fisher information for a single observation from treatment \(k \in [K]\). The condition $\mathbb{V}ar(\xi_{1,k}) = I_k(\theta_k)^{-1}$ is naturally satisfied for regular one-parameter exponential families when the parameter of interest $\theta_k$ is the mean parameter (see Proposition~\ref{prop:exp-fam} in Appendix~\ref{appnC}). Hence the minimal asymptotic variance is equal to $GVG^{\top}$, which is the asymptotic variance under any $\alpha$RTS design when conditions \hyperlink{CA}{\textbf{A}}--\hyperlink{CB}{\textbf{B}} are satisfied.
\end{remark}

Theorem~\ref{thm:normality-gen-erade} allows to perform asymptotic inference based on the data collected adaptively with our RAR procedures. Indeed, CLTs for both the estimated parameters and the proportions permit to give the asymptotic distributions of some test statistics that are transformations of these quantities (using the Delta method). In a two-armed case we are typically interested in the asymptotic calibration of a Wald test for testing the equality of two treatments. In Section~\ref{sec:simulation} we provide another example of calibration of an homogeneity test with 3 arms, see Lemma~\ref{lem:pearson}.

	\subsection{Proof Sketch}\label{sec:rts-proof}
Our analysis follows a similar template as the analysis of ERADE by \cite{hu2009efficient}. We first introduce the natural extension to $k$ arms of the quantities introduced in the analysis of ERADE:
	\begin{itemize}
		\item $p_{m,k} := \mathbb{E}[X_{m,k} \mid \mathcal{F}_{m-1}]$, the conditional probability of assigning treatment $k$ at step $m$
		\item $M_{n,k} := \sum_{m=1}^{n}(X_{m,k} - p_{m,k})$, a martingale capturing the deviation of the actual assignment from its conditional expectation
	\end{itemize}
as well as
\[U_{n,k} = \sum_{m=1}^{n-1} \alpha \hat{\rho}_{m,k} + M_{n,k} - n \hat{\rho}_{n,k}\;.\]

Our first contribution is to identify some sufficient conditions under which the conclusions of Theorem~\ref{thm:LLN} and Theorem~\ref{thm:normality-gen-erade} follow.

\begin{lemma} \label{lem:generic-conditions-main} Consider a RAR procedure under which there exists a sequence of random variables $(\ell_{n,k})_{n,k}$ satisfying the following properties:
\begin{enumerate}[label=(\roman*)]
		\item for all $k\in [K]$, \(\left(\ell_{n,k}\right)_{n \geq 1}\) is a non-decreasing sequence and for all \(n \in \mathbb{N}: \,\ell_{n,k} \leq n\)
		\item for all $k\in [K]$, $n \in \mathbb{N}$,
		\begin{equation}\label{eq:prelim-sketch}
			N_{n,k} - n\hat{\rho}_{n,k} \leq 1+N_{\ell_{n,k},k} - \ell_{n,k} \hat{\rho}_{\ell_{n,k},k} + U_{n,k} - U_{\ell_{n,k},k}
		\end{equation}
		\item for all $k\in [K]$, $n \in \mathbb{N}$,
		\begin{equation}\label{eq:condition-sketch}
			N_{\ell_{n,k},k} - \ell_{n,k} \hat{\rho}_{\ell_{n,k},k} = o(\sqrt{n}) \ a.s.
			\quad
		\end{equation}
	\end{enumerate}
Then the conclusions of Theorem~\ref{thm:LLN} and Theorem~\ref{thm:normality-gen-erade} hold.
\end{lemma}

We obtained these conditions by isolating the key arguments that are used in the ERADE analysis for $K=2$ arms, and by showing that they can actually be naturally be extended to $K$ arms. In Lemma~\ref{lem:preliminary} in Appendix~\ref{appnA}, we show that these conditions are satisfied for the random sequence
\[\ell_{n,k} := \max \left\{\, m \leq n ; \, \frac{N_{m,k}}{m} \leq \hat{\rho}_{m,k}  \,\right\}\]
for any $\alpha$RTS design. In words $\ell_{n,k}$ is the last time before \(n\) at which treatment \(k\) is under-sampled: its allocation proportion is less than or equal to the estimated target proportion.

\smallskip

We now briefly sketch the proof of Lemma~\ref{lem:generic-conditions-main}. The detailed proof that Theorem~\ref{thm:LLN} holds can be found in Appendix~\ref{proof:thm1}. The first step to establish the convergence of the empirical proportions is to show that the sequence $\hat{\rho}_n$ has a limit $u \in \Delta_K$ that is non-sparse, i.e. that satisfies $u_k \neq 0$ for all $k$. The argument there relies on both Condition \hyperlink{CA}{\textbf{A}} (to apply the law of large number) and Condition \hyperlink{CB}{\textbf{B}}.

Then, combining Inequalities~\eqref{eq:prelim-sketch} and \eqref{eq:condition-sketch} yields that for all $k$
\begin{equation}\label{step-main}\frac{N_{n,k}}{n} - \hat{\rho}_{n,k} \leq \frac{U_{n,k}- U_{\ell_{n,k},k}}{n} + o\left(\frac{1}{\sqrt{n}}\right).\end{equation}
Starting from the definition of $U_{n,k}$, we can now apply the law of large number to the martingale $M_n$ and use the convergence of $\hat{\rho}_{n}$ to prove that $\lim_{n\rightarrow \infty} \frac{U_{n,k}}{n} = -(1-\alpha)u_k \in (0,1)$. Some algebra (Lemma~\ref{lem:convergence}) further permits to show that the right-hand side of Equation~\eqref{step-main} tends to zero. Using that both $N_n/n$ and $\hat{\rho}_n$ are probability vectors, it follows that $\lim_{n\rightarrow 0} \frac{N_{n,k}}{n} - \hat{\rho}_{n,k} = 0$ and that ${N_{n}}/{n}$ also converges to $u \in (0,1)^K$.

A straightforward consequence is that $\forall k \in [K], N_{n,k} \to \infty$ a.s., hence $\hat{\Theta}_n$ converges to $\Theta$ and by continuity of \(\rho\) we obtain that \(\hat{\rho}_n = \rho(\hat{\Theta}_n) \to \rho(\Theta) = v \in (0,1)^K\), hence for all $k$,
	\[
	\lim_{n \to \infty} \frac{N_{n,k}}{n} = \lim_{n \to \infty} \hat{\rho}_{n,k} = v_k\;.
	\]
	The convergence rates follows from the Law of Iterated Logarithm.

	\smallskip

	The detailed proof that Theorem~\ref{thm:normality-gen-erade} holds can be found in Appendix~\ref{proof:thm2}. To get the more precise rates in statement (i), we need a more careful control on the increment of the auxiliary process $U_{n,k}$. Namely, we prove through a series of (rather technical) lemmas that
	\[
	U_{n,k}-U_{\ell_{n,k},k}\leq o_p(\sqrt n),
	\qquad
	U_{n,k}-U_{\ell_{n,k},k}\leq O(\sqrt{n\log\log n}) \quad a.s.
	\]
	which yields the following using the identities $\sum_{k=1}^K N_{n,k}=n$ and $\sum_{k=1}^K \hat\rho_{n,k}=1$:
	\[
	N_{n,k}-n\hat\rho_{n,k}=o_p(\sqrt n)\;.
	\]
    This proves the statements \eqref{eq:thm-res-1}--\eqref{eq:thm-res-2}. 	Combined with the rate \(n(\hat\rho_n-v)=O(\sqrt{n\log\log n})\)
	from Theorem \ref{thm:LLN} yields \eqref{eq:thm-res-3}.
	
	To establish the asymptotic normality in (ii), we prove the following lemma that gives the joint CLT of the components of the estimator \(\hat{\Theta}_n\) scaled with \(\sqrt{N_{n,k}}\). The proof of Lemma \ref{lem:componentwise} relies on a martingale-CLT argument and is given in Appendix~\ref{proof:thm2}.
	
		\begin{restatable}{lemma}{jointcltlemma}[Joint CLT with Componentwise Scaling]\label{lem:componentwise} 
			Consider a RAR procedure under which $N_{n,k}\to\infty$ a.s. for all $k \in [K]$, then we have
			\[
			D_n(\hat\Theta_n-\Theta)\xrightarrow{\mathcal{D}}
			\mathcal N\!\big(0,\mathrm{diag}\left(V_1,\dots,V_K\right)\big)
			\]
			where $D_n:=\mathrm{diag}(\sqrt{N_{n,1}},\dots,\sqrt{N_{n,K}})\;.$
		\end{restatable}
		
		Using the fact that \(\sqrt n D_n^{-1}\) converges to \(\mathrm{diag}\left(\frac{1}{\sqrt{v_1}},\ldots,\frac{1}{\sqrt{v_K}}\right)\), the CLT result of Lemma \ref{lem:componentwise} combined with Slutsky's lemma yields the first asymptotic normality result:
		\begin{align*}
			\sqrt n(\hat\Theta_n-\Theta) &= \sqrt n D_n^{-1} D_n(\hat\Theta_n-\Theta)\\
			&\xrightarrow{\mathcal{D}}\mathrm{diag}\left(\frac{1}{\sqrt{v_1}},\ldots,\frac{1}{\sqrt{v_K}}\right) \mathcal N\!\big(0,\mathrm{diag}\left(V_1,\dots,V_K\right)\big)\\ 
			&= \mathcal N(0,V).
		\end{align*}

	A first-order Taylor expansion of the target function around $\Theta$ gives
	\[
	\hat\rho_n-v
	= G(\hat\Theta_n-\Theta) + O_p(n^{-1/2}),
	\]
	and hence
	\[
	\sqrt n(\hat\rho_n-v)\xrightarrow{\mathcal{D}}\mathcal N(0,GVG^\top).
	\]
	Finally, the bound \eqref{eq:thm-res-1} implies
	\[
	\sqrt n\!\left(\frac{N_n}{n}-\hat\rho_n\right)=o_p(1),
	\]
	so Slutsky’s lemma yields the CLT for $\sqrt n(N_n/n-v)$ as well as the joint convergence stated in \eqref{eq:thm-res-4} and \eqref{eq:thm-res-5}.
	
	\begin{remark}
			We emphasize that one component of the proof, Lemma \ref{lem:componentwise}, only requires that each treatment is sampled
			infinitely often. The normalization in this CLT result is performed componentwise through
			\(D_n\), allowing each estimator \(\hat\theta_{n,k}\) to be scaled according
			to its own sample size \(N_{n,k}\). As a result, this Lemma provides a general asymptotic normality result for the treatment effect estimators that
			does not rely on any assumption on the limiting allocation proportions. This result will be particularly relevant to tackle sparse target allocations, as we do in the next section.
	\end{remark}

	\section{$\alpha$RTS with Forced Exploration ($\alpha$RTS-FE)}\label{sec:alpha-rts-fe}
	
The adaptive design framework introduced above is augmented with a forced-exploration mechanism ensuring that each treatment continues to receive sufficient sampling throughout the experiment while the adaptive targeting rule remains active. Such a mechanism facilitate asymptotic analysis and prevent treatment starvation, while avoiding rigid allocation structures that do not adapt to the evolving information collected during the trial.

The forced-exploration component dynamically allocates observations to under-sampled
treatments in a randomized manner, ensuring adequate information acquisition without
compromising the adaptivity of the design. This approach is inspired by tracking-based
strategies from the bandit best-arm identification literature, notably the forced-exploration
component of the D-Tracking algorithm of \cite{garivier2016optimal}. By operating
alongside the targeting rule, forced exploration provides a natural mechanism to
guarantee that $N_{n,k}\to\infty$ for all treatments $k$, including in settings where
the target allocation may be sparse.
 
\begin{definition}\label{def:rts-fe-def}
	Let $h:\mathbb{N}\to\mathbb{R}$ a function such that
	\begin{enumerate}[label=(\roman*)]
		\item $\lim\limits_{n \to +\infty} h(n) = +\infty$
		\item $h(n)=o(\sqrt{n})$ as \(n \to +\infty\)
	\end{enumerate}

%

	Also we define the following sets
	\begin{equation}
		\mathbf{U}_m:=\{j\in[K]:\ N_{m,j}\le h(m)\},\qquad \mathbf{S}_m:=\arg\min_{j\in U_m} N_{m,j}.
	\end{equation}
	A design is in the \(\alpha\)RTS-FE family if its allocation probabilities obey:
	\begin{align}
		&\textbf{(Forced exploration)}\quad \mathbf{U}_m\neq\emptyset\ \Rightarrow\ \forall a\in \mathbf{S}_m:\ p_{m+1,a}=\frac{1}{|\mathbf{S}_m|};\label{eq:fe-phase}\\
		&\textbf{(Otherwise)}\quad \mathbf{U}_m=\emptyset\ \Rightarrow\ \text{$p_{m+1}$ satisfies \eqref{eq:throttle} for some $\alpha\in[0,1)$}.\label{eq:fe-otherwise}
	\end{align}
\end{definition}

\begin{example}
	In \cite{garivier2016optimal}, the \textit{D-Tracking} algorithm is composed of two parts, the \textit{Tracking} part as presented earlier, and the second part consists in \textit{Forced Exploration} that coincides with the \textit{\(\alpha\)RTS-FE} family with \(\alpha=0\), where they chose \(h(n) = \left(\sqrt{n} - \frac{K}{2}\right)^+.\) Note that this function doesn't verify the second condition of Definition \ref{def:rts-fe-def}. In our experiments, we use instead the function \(h(n) = \left(n^{\frac{1}{3}} - \frac{K}{2}\right)^+\) that satisfies the two conditions.
\end{example}

\subsection{Preservation of Asymptotic Properties}

We first show that the introduction of forced exploration does not alter the asymptotic behavior established for the \(\alpha RTS\) family for non-sparse target allocations, i.e. target allocations that satisfy {Condition} \hyperlink{CB}{\textbf{B}}.

\begin{theorem}[Validity under forced exploration]\label{thm:forced} Under conditions \hyperlink{CA}{\textbf{A}}--\hyperlink{CB}{\textbf{B}}, Theorem~\ref{thm:LLN} and Theorem~\ref{thm:normality-gen-erade} 
	remain valid for any design in the \(\alpha\)RTS-FE family.
\end{theorem}

\begin{proof} The proof follow the exact same template as that for $\alpha$RTS, with the twist of introducing an alternative sequence $\ell_{n,k}$ defined as follow:
\begin{equation*}
	\ell_{n,k} := \max \left\{\, m \leq n : \frac{N_{m,k}}{m} \leq \hat{\rho}_{m,k} \text{ or } \left(k \in \text{arg}\!\min_{j \in U_m} N_{m,j}\text{ and } \mathbf{U}_m \neq \emptyset\right) \;.\right\}
\end{equation*}
$\ell_{n,k}$ now represents the last time before \(n\) at which treatment \(k\) is under-sampled or is the least sampled among treatments belonging to \(U_{\ell_{n,k}}\). We prove in Lemma~\ref{lem:preliminary2} in Appendix~\ref{appnB} that the conditions in Lemma~\ref{lem:generic-conditions-main} are valid for this new sequence $\ell_{n,k}$. The conclusion follows. The proof of Lemma~\ref{lem:preliminary2} hinges on some properties on the function \(h\) to establish that
\begin{align*}
	N_{\ell_{n,k},k} - \ell_{n,k} \hat{\rho}_{\ell_{n,k},k}
	&\leq h(\ell_{n,k})
	= o(\sqrt{n})
	\: .
\end{align*}
\end{proof}

\subsection{Sparse Allocation Target} 

In this section we relax Condition \hyperlink{CB}{\textbf{B}} to incorporate possibly sparse target allocations, in the following sense.

\begin{definition}[Sparse Target Allocation]
	A target allocation function \(\rho: \mathbb{R}^K \to [0,1]^K\) is said to be sparse at \(\Theta\) if one or more components of the target vector \(v = \rho(\Theta)\) are equal to zero.
\end{definition}

We now establish consistency and convergence rates for \(\alpha\)RTS-FE design for sparse target allocation. The proof of all statements in this section can be found in Appendix~\ref{proofs:sparse-case}.


\begin{theorem}[Consistency and rate]\label{thm:sparse-rate}
	Under condition \hyperlink{CA}{\textbf{A}}, \(\alpha\)RTS-FE satisfies as \(n \to +\infty\), for all $k \in [K]$,
	\[
	N_{n,k}\xrightarrow{\text{a.s.}}+\infty,\qquad
	\frac{N_{n,k}}{n}\xrightarrow{\text{a.s.}}v_k \ \ \ \
	\text{ and } \ \ \ \ \hat\Theta_n-\Theta=O\!\left(\sqrt{\frac{\log\log N_{n,k}}{N_{n,k}}}\right)\text{ a.s.}.
	\]
\end{theorem}

Finally, we derive a joint componentwise CLT for the estimators when adaptive data is generated from an \(\alpha\)RTS-FE design, which is a direct consequence of Lemma \ref{lem:componentwise}.

\begin{corollary}\label{corr:componentwise}
		Under condition \hyperlink{CA}{\textbf{A}}, \(\alpha\)RTS-FE satisfies as \(n \to +\infty\)
		\[
		D_n(\hat\Theta_n-\Theta)\xrightarrow{\mathcal{D}}
		\mathcal N\!\big(0,\mathrm{diag}\left(V_1,\dots,V_K\right)\big)
		\]
		with $D_n:=\mathrm{diag}(\sqrt{N_{n,1}},\dots,\sqrt{N_{n,K}})$
	\end{corollary}
	
\begin{remark}[CLTs under sparsity]\label{cor:sparse-clt}
		The normalization in Corollary~\ref{corr:componentwise} is component-wise and depends on \(N_{n,k}\), rather than on \(n\). This formulation naturally accommodates sparse targets, for which
		\(v_k = 0\) for some $k$. We note that when \(v \in (0, 1)^K\), this result recovers some known asymptotic normality results established for response-adaptive procedures that converge almost surely to their targets (Lemma 1 from  \cite{hu2006asymptotically}).
	\end{remark}

	\section{Simulation study}\label{sec:simulation}
	
We now complement our asymptotic analysis with a simulation study designed to evaluate the behavior of $\alpha$RTS designs. The implementation and simulation scripts used in this study can be accessed at \url{https://github.com/ryagouti/rebalancing-targeting-strategies}.

	Our experiments address three questions. First, we investigate the
	convergence of allocation proportions and plug-in target estimates in a
	three-treatment setting (\(K=3\)). Second, we study the impact of forced
	exploration when the target allocation is sparse. Finally, we evaluate the
	performance of hypothesis testing procedures based on data collected with
	\(\alpha\)RTS and \(\alpha\)RTS-FE designs in a four-treatment setting
	(\(K=4\)). Together, these experiments illustrate both the allocation
	properties and the inferential performance of the proposed methods.
	
	We set \(\alpha=0.4\) throughout the experiments. This choice follows the recommendation of \cite{hu2009efficient}, who suggest values of \(\alpha\) between \(0.4\) and \(0.7\)
	based on their numerical study of ERADE.

\subsection{Convergence of Allocation Proportions and Plugin Estimate}

The first experiment is designed to illustrate the convergence behavior predicted by Theorems \ref{thm:LLN} and \ref{thm:normality-gen-erade}. We consider a three-treatment setting (\(K=3\)) with Bernoulli arms with success probabilities $\Theta=(0.5, 0.6, 0.8)$, and use the Neyman allocation as the target, defined as
\[\rho_k^N(\Theta) = \frac{\sqrt{\theta_k(1-\theta_k)}}{\sum_{\ell=1}^{K}\sqrt{\theta_{\ell}(1-\theta_{\ell})}}\;.\]
For two-arms the Neyman allocation (\cite{neyman1934representative}) is the one minimizing the variance of the estimator of the treatment difference as shown in \cite{melfi1998variablility}, and its natural extension to $K$ arms can be traced back to \cite{Cochran1977}.

Figures \ref{fig:fixed-tgt-prop} and \ref{fig:fixed-tgt-plugin} respectively  display  the evolution of the absolute errors $|N_{n,k}/n - v_k|$ and $|\hat{\rho}_{n,k} - v_k|$ as a function of $n$, for the three example of $\alpha$RTS design discussed in Section~\ref{sec:alpha-rts}: Distance-Based, ERADE 2025 and Interpolated D-Tracking. Across all treatments, both quantities decrease rapidly as the sample size increases, indicating that the allocation proportions track the target allocation and that the plug-in estimator stabilizes. This behavior is consistent with the almost sure convergence results  implied by Theorems \ref{thm:LLN} and \ref{thm:normality-gen-erade}. While all three designs share the same asymptotic guarantees, we observe minor differences in early-stage variability. The D-Tracking rule exhibits larger fluctuations at small sample sizes, due to its more aggressive reallocation strategy, whereas the distance-based and ERADE-type rules yield smoother trajectories. These differences vanish as $n$ increases, and all methods display similar behavior in later stages.

\begin{figure}[tb]
	\centering
	\begin{subfigure}[t]{0.32\linewidth}
		\centering
		\includegraphics[width=\linewidth]{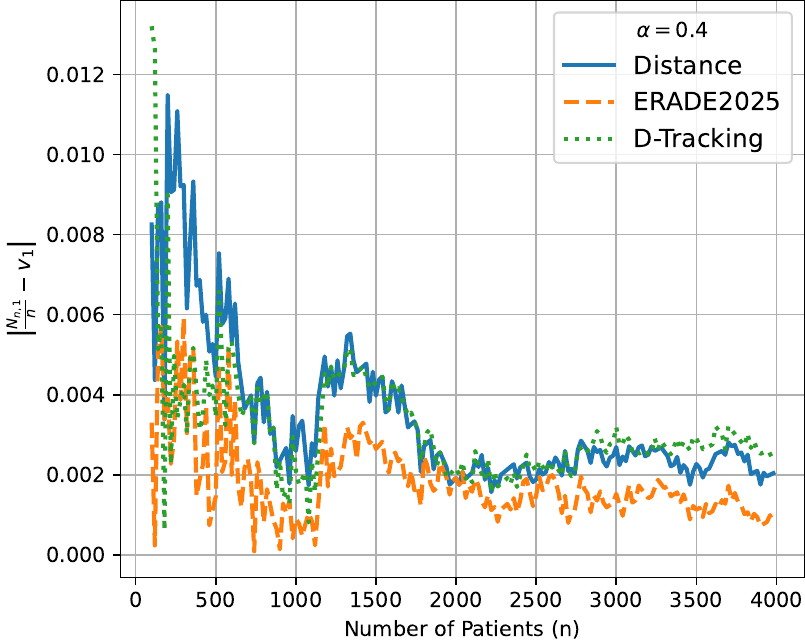}
		\caption{Treatment 1}
	\end{subfigure}
	\hfill
	\begin{subfigure}[t]{0.32\linewidth}
		\centering
		\includegraphics[width=\linewidth]{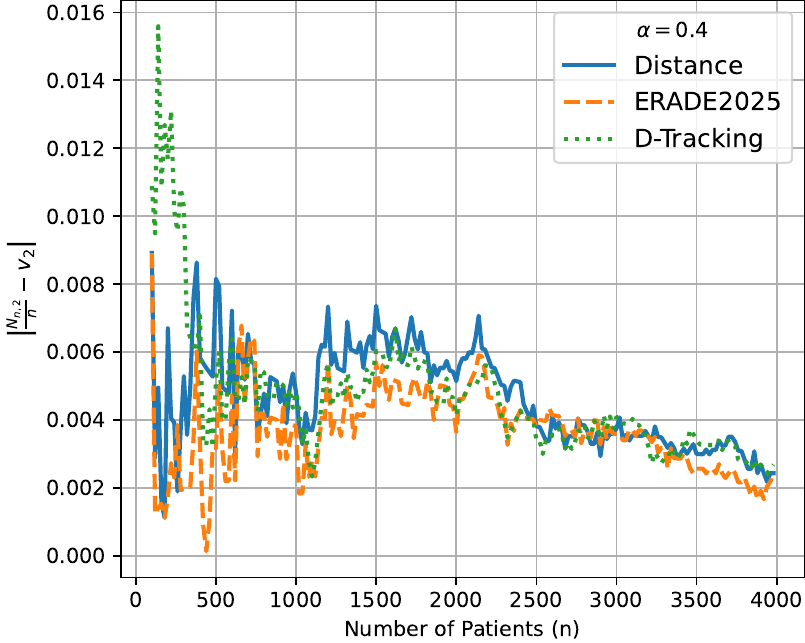}
		\caption{Treatment 2}
	\end{subfigure}
	\hfill
	\begin{subfigure}[t]{0.32\linewidth}
		\centering
		\includegraphics[width=\linewidth]{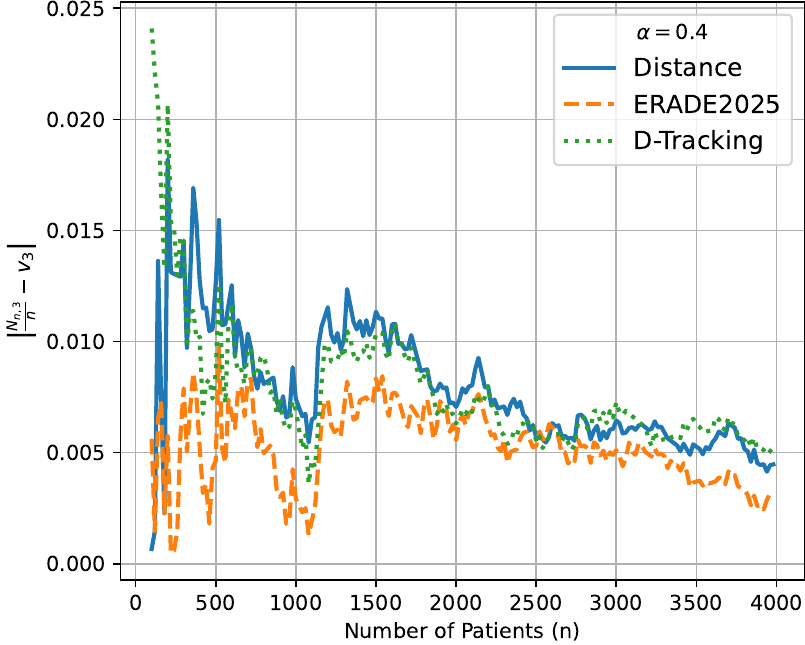}
		\caption{Treatment 3}
	\end{subfigure}
	
	\caption{Distance between the empirical allocation proportion and the target as a function of $n$ for different $\alpha$RTS designs targeting the Neyman Allocation}
	\label{fig:fixed-tgt-prop}
\end{figure}

\begin{figure}[tb]
	\centering
	\begin{subfigure}[t]{0.32\linewidth}
		\centering
		\includegraphics[width=\linewidth]{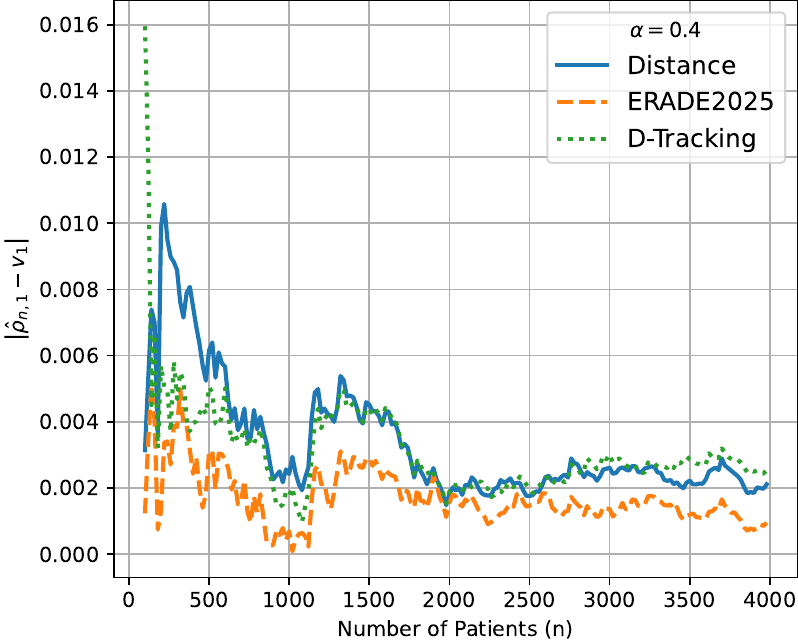}
		\caption{Treatment 1}
	\end{subfigure}
	\hfill
	\begin{subfigure}[t]{0.32\linewidth}
		\centering
		\includegraphics[width=\linewidth]{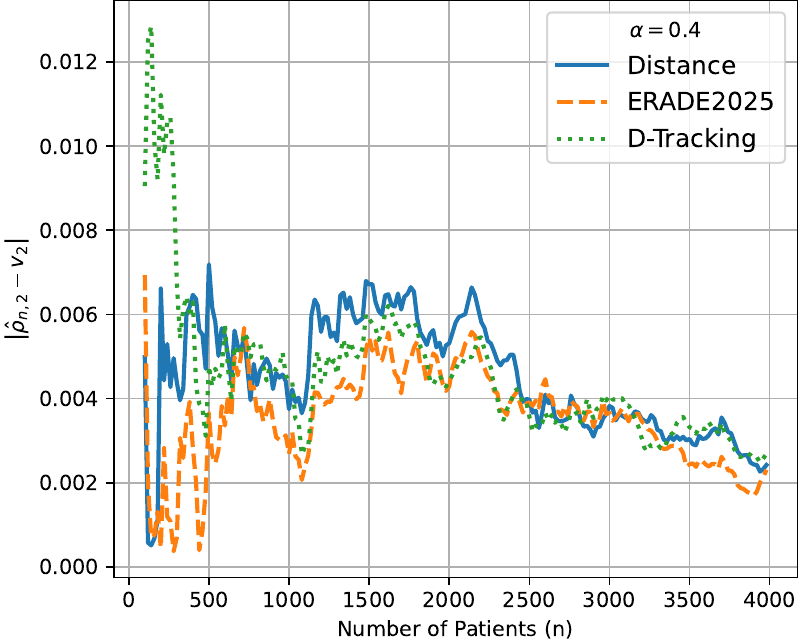}
		\caption{Treatment 2}
	\end{subfigure}
	\hfill
	\begin{subfigure}[t]{0.32\linewidth}
		\centering
		\includegraphics[width=\linewidth]{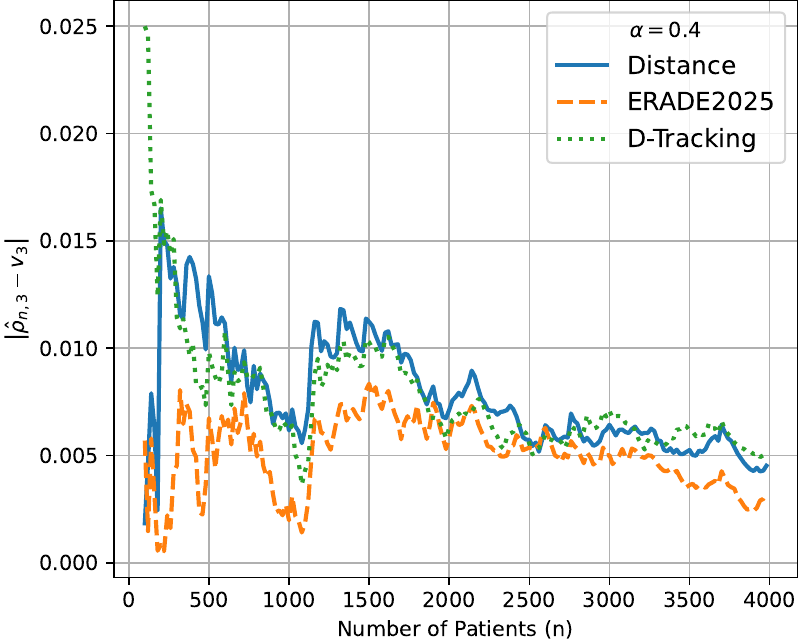}
		\caption{Treatment 3}
	\end{subfigure}
	
	\caption{Distance between the target plug-in estimate and the target as a function of $n$ for different $\alpha$RTS designs targeting the Neyman Allocation}
	\label{fig:fixed-tgt-plugin}
\end{figure}

\subsection{The Impact of Forced Exploration} Our second experiments aims at illustrating the impact of forced exploration when the target allocation is sparse.
The work of \cite{tymofyeyev2007implementing} proposes a family of target allocations in multi-treatment Bernoulli experiments that are motivated by homogeneity tests, that is, testing
the equality of treatment success probabilities across all arms. In this work, we consider the allocation that minimizes the sample size given a constraint on the power of the homogeneity test (which can be related to a constraint on the non-centrality parameter of some chi-square distribution). This target allocation is shown to exhibit a sparsity pattern for a wide range of parameter configurations. In particular, if $\Theta$ is such that
\[\theta_1=\cdots=\theta_s
>
\theta_{s+1}\geq\cdots\geq \theta_{K-g}
>
\theta_{K-g+1}=\cdots=\theta_K ,\]
then
\[
\rho_i^T(\Theta)=
\begin{cases}
	\displaystyle
	\frac{\sqrt{\theta_1 (1 - \theta_1)}}
	{s\left(\sqrt{\theta_1 (1 - \theta_1)}+\sqrt{\theta_K (1 - \theta_K)}\right)},
	& \text{if } i\in\{1,\dots,s\},
	\\[1.2em]
	0,
	& \text{if } i\in\{s+1,\dots,K-g\},
	\\[1.2em]
	\displaystyle
	\frac{1}{g}\left(
	1-
	\frac{\sqrt{p_1q_1}}
	{\sqrt{p_1q_1}+\sqrt{p_Kq_K}}
	\right),
	& \text{if } i\in\{K-g+1,\dots,K\}.
\end{cases}
\]
and \(\rho_{s+1}^T=\cdots=\rho_{K-g}^T=0\). This allocation, which we call the Tymofyeyev allocation in our study, assigns positive target proportions only to the best-performing and the worst-performing treatments. 


In our experiments, we consider a trial with $K=3$ arms and binary responses. The success probabilities are chosen from the following configurations
\[\Theta=(\theta_1,\theta_2,\theta_3) \in \{(0.1,0.3,0.6), (0.3,0.4,0.7), (0.4,0.6,0.8)\}\]
which induce sparse Tymofyeyev target allocations with treatment 2 being intermediate and asymptotically eliminated. The target values \(\rho^{T}\left(\Theta\right)\) for these instances are reported  in Table~\ref{tab:rho_tymo}.

\begin{table}[b]
	\centering
	\begin{tabular}{c|ccc}
		\hline
		$(\theta_1,\theta_2,\theta_3)$ & Arm 1 & Arm 2 & Arm 3 \\
		\hline
		$(0.1,\,0.3,\,0.6)$ & 0.3798 & 0 & 0.6202 \\
		$(0.3,\,0.4,\,0.7)$ & 0.5000 & 0 & 0.5000 \\
		$(0.4,\,0.6,\,0.8)$ & 0.5505 & 0 & 0.4495 \\
		\hline
	\end{tabular}
	\caption{\label{tab:rho_tymo}Tymofyeyev target allocation vectors corresponding to different instances $\Theta$}
\end{table}

We compare \(\alpha\)RTS and its \(\alpha\)RTS-FE counterpart for the same three variants considered in the previous section, when the target is a sparse Tymofyeyev allocation.
For each configuration we perform 500 independent runs of each algorithms, with a total of \(n=1000\) patients per run. Table \ref{tab:sparse-target} reports the average empirical allocation proportions $N_{n,k}/n$ obtained for different variants of $\alpha$RTS.
Across all instantiations (Distance, ERADE2025 and D-Tracking) and values of $\Theta$, the standard \(\alpha\)RTS procedures allocate a non-negligible fraction of observations to the intermediate treatment, despite its target allocation being zero. The forced-exploration variants (\(\alpha\)RTS-FE) tend to assign a slightly smaller proportion of observations to this treatment, resulting in empirical allocations that are closer to the target. 

\begin{table}[tb]
	\centering
	\small
	\begin{tabular}{c|ccc|ccc}
		\hline
		\textbf{Distance} & \multicolumn{3}{c|}{\(\alpha.RTS\)} & \multicolumn{3}{c}{\(\alpha.RTS-FE\)} \\
		$(p_1, p_2, p_3)$
		& $\frac{N_1}{n}$ & $\frac{N_2}{n}$ & $\frac{N_3}{n}$
		& $\frac{N_1}{n}$ & $\frac{N_2}{n}$ & $\frac{N_3}{n}$ \\
		\hline
		(0.1, 0.3, 0.6) & 0.3381 & 0.0849 & 0.5760 & 0.3559 & 0.0558 & 0.5873 \\
		(0.3, 0.4, 0.7) & 0.3634 & 0.1523 & 0.4832 & 0.3814 & 0.1387 & 0.4790 \\
		(0.4, 0.6, 0.8) & 0.4766 & 0.1206 & 0.4017 & 0.4969 & 0.0867 & 0.4154 \\
		\hline
	\end{tabular}
	
\vspace{0.4cm}

	\begin{tabular}{c|ccc|ccc}
		\hline
		\textbf{ERADE2025} & \multicolumn{3}{c|}{\(\alpha.RTS\)} & \multicolumn{3}{c}{\(\alpha.RTS-FE\)} \\
		$(p_1, p_2, p_3)$
		& $\frac{N_1}{n}$ & $\frac{N_2}{n}$ & $\frac{N_3}{n}$
		& $\frac{N_1}{n}$ & $\frac{N_2}{n}$ & $\frac{N_3}{n}$ \\
		\hline
		(0.1, 0.3, 0.6) & 0.3422 & 0.0968 & 0.5600 & 0.3544 & 0.0480 & 0.5966 \\
		(0.3, 0.4, 0.7) & 0.3529 & 0.1685 & 0.4776 & 0.3802 & 0.1350 & 0.4838 \\
		(0.4, 0.6, 0.8) & 0.4873 & 0.1012 & 0.4105 & 0.4900 & 0.0938 & 0.4151 \\
		\hline
	\end{tabular}
	
\vspace{0.4cm}

	\begin{tabular}{c|ccc|ccc}
		\hline
		\textbf{D-Tracking} & \multicolumn{3}{c|}{\(\alpha.RTS\)} & \multicolumn{3}{c}{\(\alpha.RTS-FE\)} \\
		$(p_1, p_2, p_3)$
		& $\frac{N_1}{n}$ & $\frac{N_2}{n}$ & $\frac{N_3}{n}$
		& $\frac{N_1}{n}$ & $\frac{N_2}{n}$ & $\frac{N_3}{n}$ \\
		\hline
		(0.1, 0.3, 0.6) & 0.3375 & 0.0979 & 0.5636 & 0.3593 & 0.0499 & 0.5898 \\
		(0.3, 0.4, 0.7) & 0.3747 & 0.1434 & 0.4809 & 0.3841 & 0.1360 & 0.4789 \\
		(0.4, 0.6, 0.8) & 0.4933 & 0.1003 & 0.4054 & 0.4962 & 0.0887 & 0.4141 \\
		\hline
	\end{tabular}
	
\vspace{0.4cm}

	\caption{Simulated allocation proportions for different targeting designs with fixed target allocation (Tymofyeyev), 500 simulations, n = 1000}
	\label{tab:sparse-target}
\end{table}

%

\begin{figure}[tb]
	\centering
	\begin{subfigure}[t]{0.32\linewidth}
		\centering
		\includegraphics[width=\linewidth]{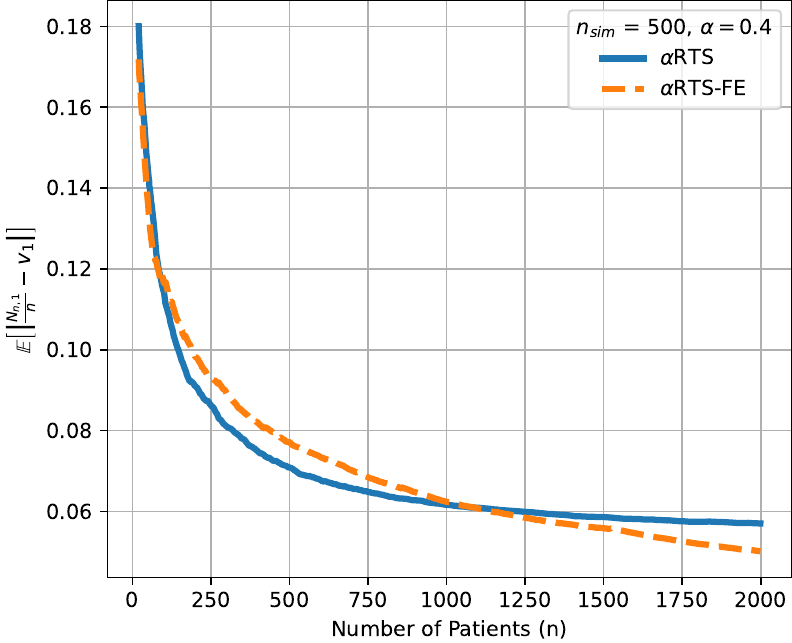}
		\caption{Treatment 1}
	\end{subfigure}
	\hfill
	\begin{subfigure}[t]{0.32\linewidth}
		\centering
		\includegraphics[width=\linewidth]{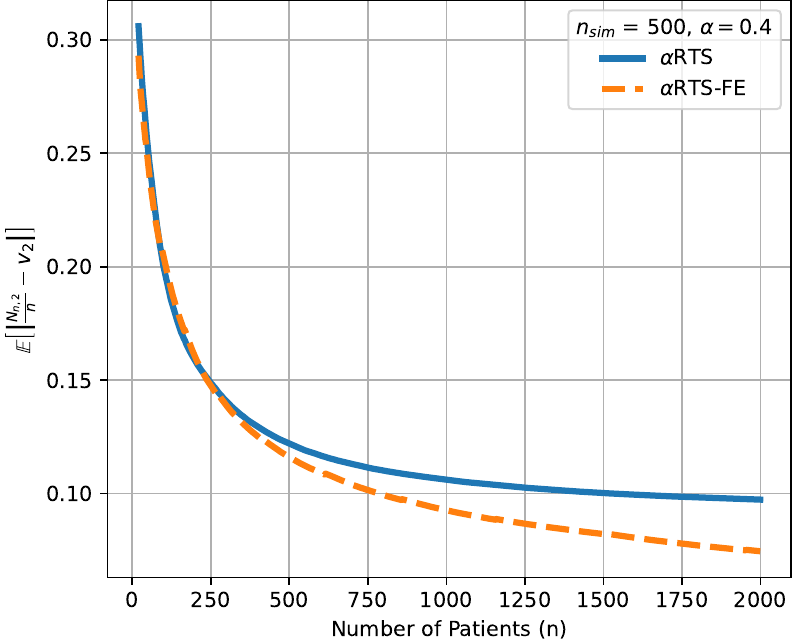}
		\caption{Treatment 2}
	\end{subfigure}
	\hfill
	\begin{subfigure}[t]{0.32\linewidth}
		\centering
		\includegraphics[width=\linewidth]{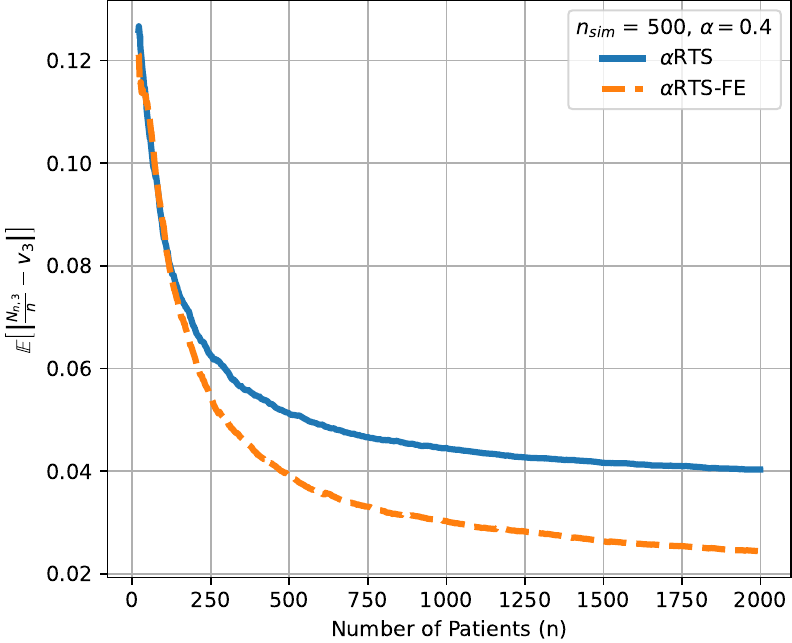}
		\caption{Treatment 3}
	\end{subfigure}

	\caption{Distance between the allocation proportion and the target as a function of $n$ based on Distance design and the Tymofyeyev allocation}
	\label{fig:tymo-prop}
\end{figure}

To complement these findings, Figure \ref{fig:tymo-prop} reports the average allocation error
\(\left|\frac{N_{n,k}}{n}-v_k\right|\) as a function of the number of patients $n$, computed over \(n_{\mathrm{sim}}=500\) independent simulation runs, using the real parameters \(\Theta = (0.4, 0.6, 0.8)\) and the Distance targeting design with and without forced exploration. Across all
treatments, the error decreases with the sample size for both \(\alpha\)RTS and
\(\alpha\)RTS-FE. In the sparse setting considered here, the forced-exploration
variants tend to achieve smaller errors. This suggests that continued sampling of all
treatments lead to a faster convergence to a sparse target allocation.

\subsection{Hypothesis Testing with $\alpha$RTS} The asymptotic properties established for $\alpha$RTS(-FE) make these adaptive data collection methods compatible with asymptotic inference. To complement these asymptotic results, our last set of experiments focuses on the finite-time performance of an hypothesis test performed when the data is collected with $\alpha$RTS. We consider $K$ Bernoulli arms and the homogeneity test already mentioned above:
\[\mathcal{H}_0 : \{\forall k \in [K] : \theta_k = \theta_1 \} \ \ \text{ against } \ \ \ \mathcal{H}_1 : \{\exists i,j :  \theta_i \neq \theta_j\}\;.\]
The following lemma leads to an asymptotic calibration of a chi-square test.

\begin{lemma}\label{lem:pearson}
		Define the quantities \(S_{n,k} := N_{n,k} \hat{\theta}_{n,k}\) and 
		\(\tilde{\theta}_n := \sum_{k=1}^K \frac{N_{n,k}}{n} \hat{\theta}_{n,k}\).	
		Consider a RAR procedure under which \(N_{n,k} \to +\infty\), then the Pearson Chi-square statistic \(\mathcal{X}_n := \sum_{k=1}^K \frac{\left(S_{n,k} - N_{n,k}\tilde{\theta}_n\right)^2}{N_{n,k}\tilde{\theta}_n (1 - \tilde{\theta}_n)}\)
		satisfies
		\(\mathcal{X}_n \xrightarrow{\mathcal{D}} \chi^2_{K-1}\) under the null hypothesis.
	\end{lemma}
	\begin{remark}
		Lemma \ref{lem:pearson} is valid for all \(\alpha\)RTS designs when Condition \hyperlink{CB}{\textbf{B}} is verified, and also for \(\alpha\)RTS-FE designs.
\end{remark}

{Hence, the test that rejects $\mathcal{H}_0$ if $\mathcal{X}_n$ exceeds the quantile of order 0.95 of a chi-square distribution with $K-1$ degrees of freedom is asymptotically a level-$0.05$ test.} We now evaluate its finite-sample performance for different targeting strategies in the $\alpha$RTS family and different target allocations, on some instances with $K=4$ arms.

\paragraph*{Varying the Targeting Mechanism with Fixed Target}
We first evaluate the statistical performance of the proposed designs when the target allocation is fixed to the Neyman allocation. In this setting, the objective is to assess whether different targeting mechanisms within the $\alpha$RTS and $\alpha$RTS-FE families preserve the validity of hypothesis testing procedures. We also compare those adaptive data collection procedures to non-adaptive uniform sampling. 

Table \ref{tab:fixed-target} reports the empirical type I error rates of the Pearson Chi-squared test under the null hypothesis for a range of homogeneous configurations, along with power values under configurations in the alternative. All results are obtained from $1000$ independent simulations with a sample size $n=200$.


Under the null hypothesis (top panel of Table~\ref{tab:fixed-target}),
	all designs maintain type~I error rates reasonably close to the nominal
	\(5\%\) level, supporting the validity of the asymptotic calibration even
	for the moderate sample size \(n=200\). Although some variability is
	observed across parameter configurations and targeting rules, no design
	appears systematically best.
	
	Under the alternative, all methods achieve high power for well-separated configurations, making differences between targeting rules difficult to discern. More informative comparisons arise for the challenging instances reported in the bottom panel, where the power remains substantially below one. In these settings, adaptive targeting occasionally yields modest
	improvements over uniform sampling, but no targeting rule consistently outperforms the others. Similar conclusions hold for the forced-exploration variants (\(\alpha\)RTS-FE), indicating that the
	additional exploration does not materially affect testing performance in this non-sparse regime.	Overall, these results suggest that, under a fixed Neyman target allocation, the choice of targeting mechanism has only a limited impact on the finite-sample performance of the Pearson chi-square test

\begin{table}[tb]
	\centering
	\small
	\resizebox{\linewidth}{!}{%
	\begin{tabular}{c|cccc|cccc}
		\hline
		& \multicolumn{4}{c|}{\(\alpha.RTS\)} & \multicolumn{4}{c}{\(\alpha.RTS-FE\)} \\
		$(p_0, p_1, p_2, p_3)$
		& Uniform & Distance & ERADE2025 & D-Tracking
		& Uniform & Distance & ERADE2025 & D-Tracking \\
		\hline
		(0.1, 0.1, 0.1, 0.1) & 2.80\% & 1.00\% & 1.70\% & 0.90\% & 4.00\% & 1.70\% & 1.70\% & 1.50\% \\
		(0.3, 0.3, 0.3, 0.3) & 4.30\% & 5.50\% & 3.80\% & 4.30\% & 4.90\% & 5.50\% & 4.70\% & 4.70\% \\
		(0.5, 0.5, 0.5, 0.5) & 4.00\% & 4.90\% & 5.00\% & 4.00\% & 5.50\% & 5.30\% & 5.80\% & 6.20\% \\
		(0.7, 0.7, 0.7, 0.7) & 4.30\% & 5.20\% & 3.80\% & 4.50\% & 4.90\% & 5.50\% & 4.70\% & 4.70\% \\
		(0.9, 0.9, 0.9, 0.9) & 2.80\% & 1.00\% & 1.70\% & 1.00\% & 4.00\% & 1.70\% & 1.70\% & 1.50\% \\
		\hline
		(0.9, 0.8, 0.7, 0.6) & \textbf{87.70\%} & 85.90\% & 85.50\% & 85.30\% & \textbf{87.50\%} & 84.70\% & 83.80\% & 86.50\% \\
		(0.8, 0.7, 0.5, 0.4) & 98.10\% & \textbf{98.40\%} & 97.70\% & 98.20\% & 98.10\% & \textbf{98.40\%} & 98.20\% & 98.10\% \\
		(0.7, 0.6, 0.4, 0.2) & 99.90\% & 99.90\% & 99.70\% & \textbf{100.00\%} & \textbf{99.90\%} & 99.70\% & 99.80\% & 99.60\% \\
		(0.5, 0.4, 0.3, 0.2) & 78.50\% & \textbf{81.60\%} & \textbf{81.60\%} & 81.50\% & \textbf{80.20\%} & 79.80\% & 79.20\% & 78.30\% \\
		(0.4, 0.3, 0.2, 0.1) & \textbf{86.40\%} & 83.90\% & 84.70\% & 85.20\% & \textbf{88.90\%} & 85.20\% & 85.20\% & 83.80\% \\
		\hline
		(0.5, 0.5, 0.5, 0.55) & 6.90\% & \textbf{9.00\%} & 7.80\% & 8.70\% & \textbf{9.10\%} & 8.60\% & 8.80\% & 8.90\% \\
		(0.5, 0.5, 0.55, 0.55) & 8.80\% & 8.90\% & 9.40\% & \textbf{9.70\%} & 8.20\% & \textbf{10.10\%} & 9.60\% & 9.70\% \\
		(0.5, 0.55, 0.55, 0.55) & 7.80\% & 7.90\% & 8.90\% & \textbf{9.70\%} & \textbf{8.70\%} & 8.10\% & 7.70\% & 7.00\% \\
		(0.5, 0.5, 0.55, 0.6) & 18.40\% & \textbf{20.20\%} & 18.90\% & 19.30\% & 17.10\% & \textbf{18.10\%} & 17.70\% & 17.50\% \\
		(0.5, 0.55, 0.55, 0.6) & 15.10\% & \textbf{15.70\%} & 14.60\% & 15.10\% & 14.20\% & \textbf{15.20\%} & 14.20\% & 13.40\% \\
		\hline
	\end{tabular}
}
	\caption{Type-1 error (\%) for different targeting designs with fixed target allocation (Neyman), 1000 simulations, n = 200}
	\label{tab:fixed-target}
\end{table}



\paragraph*{Varying the Target Allocations With Fixed Targeting}

We now investigate the impact of the choice of the target allocation function on statistical performance, while keeping the targeting mechanism fixed. In this experiment, we use the distance-based design and compare three target allocations: Tymofyeyev, Neyman, and RSIHR, that we also compare with the non-adaptive uniform sampling. RSIHR was proposed by
\cite{Rosenberger01Opt}. It balances ethical and statistical considerations by
minimizing the expected number of treatment failures subject to a constraint
on the variance of the test statistic and is equal to
\[\rho_k^R(\Theta) = \frac{\sqrt{\theta_k}}{\sum_{\ell=1}^{K}\sqrt{\theta_{\ell}}}\;.\]

Table \ref{tab:fixed-targeting} reports the empirical type I error rates (top panel) and power (bottom panels) of the Pearson Chi-squared test across different configurations, based on $1000$ simulations with a sample size $n=200$.

%

Under the null hypothesis, the choice of target allocation has a noticeable
	impact on type~I error control. The Tymofyeyev allocation tends to be more
	conservative than both the uniform design and the Neyman and RSIHR
	allocations, which generally yield rejection rates closer to the nominal
	\(5\%\) level. Similar behavior is observed for the
	forced-exploration variants (\(\alpha\)RTS-FE).
	
	Under the alternative, all procedures achieve high power for well-separated configurations, making differences difficult to distinguish. More informative comparisons arise for the challenging instances reported in the
	bottom panel. In these settings, Neyman and RSIHR generally outperform the Tymofyeyev allocation and occasionally improve upon uniform sampling.
	However, the gains over the uniform design remain moderate, and no target allocation exhibits a clear and systematic advantage across all configurations. The forced-exploration variants display similar trends, suggesting that the exploration mechanism has only a limited effect on testing performance in
	this regime. Overall, the results indicate that the choice of target allocation can influence both type~I error control and power, although the
	improvements over uniform sampling remain relatively modest in the configurations considered here.

\begin{table}[tb]
	\centering
	\small
	\resizebox{\linewidth}{!}{%
	\begin{tabular}{c|cccc|cccc}
		\hline
		& \multicolumn{4}{c|}{\(\alpha.RTS\)} & \multicolumn{4}{c}{\(\alpha.RTS-FE\)} \\
		$(p_0, p_1, p_2, p_3)$
		& Uniform & Tymofyeyev & Neyman & RSIHR
		& Uniform & Tymofyeyev & Neyman & RSIHR \\
		\hline
		(0.1, 0.1, 0.1, 0.1) & 2.80\% & 0.40\% & 1.00\% & 1.30\% & 4.00\% & 0.30\% & 1.50\% & 1.10\% \\
		(0.3, 0.3, 0.3, 0.3) & 4.30\% & 1.60\% & 4.70\% & 5.50\% & 4.90\% & 1.80\% & 5.10\% & 4.10\% \\
		(0.5, 0.5, 0.5, 0.5) & 4.00\% & 1.60\% & 4.70\% & 5.30\% & 5.50\% & 2.20\% & 6.30\% & 6.00\% \\
		(0.7, 0.7, 0.7, 0.7) & 4.30\% & 1.60\% & 4.80\% & 4.50\% & 4.90\% & 1.80\% & 5.10\% & 4.50\% \\
		(0.9, 0.9, 0.9, 0.9) & 2.80\% & 0.40\% & 1.10\% & 4.30\% & 4.00\% & 0.30\% & 1.50\% & 2.90\% \\
		\hline
		(0.9, 0.8, 0.7, 0.6) & 87.70\% & 81.40\% & 84.80\% & \textbf{90.30\%} & 87.50\% & 81.00\% & 84.80\% & \textbf{89.70\%} \\
		(0.8, 0.7, 0.5, 0.4) & 98.10\% & 95.40\% & 98.10\% & \textbf{98.20\%} & 98.10\% & 95.30\% & \textbf{98.70\%} & 98.50\% \\
		(0.7, 0.6, 0.4, 0.2) & \textbf{99.90\%} & 97.90\% & 99.80\% & 99.70\% & \textbf{99.90\%} & 98.50\% & \textbf{99.90\%} & 99.70\% \\
		(0.5, 0.4, 0.3, 0.2) & 78.50\% & 74.00\% & \textbf{80.60\%} & 79.20\% & 80.20\% & 75.00\% & \textbf{80.30\%} & 79.70\% \\
		(0.4, 0.3, 0.2, 0.1) & \textbf{86.40\%} & 81.70\% & 85.30\% & 83.70\% & \textbf{88.90\%} & 84.20\% & 85.60\% & 83.50\% \\
		\hline		
		(0.5, 0.5, 0.5, 0.55) & 6.90\% & 4.30\% & \textbf{7.40\%} & 6.10\% & \textbf{9.10\%} & 2.50\% & 7.70\% & 6.10\% \\
		(0.5, 0.5, 0.55, 0.55) & 8.80\% & 4.10\% & 8.10\% & \textbf{9.60\%} & \textbf{8.20\%} & 3.20\% & 6.90\% & 7.20\% \\
		(0.5, 0.55, 0.55, 0.55) & \textbf{7.80\%} & 4.10\% & \textbf{7.80\%} & 6.60\% & \textbf{8.70\%} & 2.90\% & 7.10\% & 7.40\% \\
		(0.5, 0.5, 0.55, 0.6) & \textbf{18.40\%} & 9.30\% & 16.10\% & 14.70\% & \textbf{17.10\%} & 8.00\% & 13.90\% & 12.40\% \\
		(0.5, 0.55, 0.55, 0.6) & \textbf{15.10\%} & 7.60\% & 13.20\% & 11.20\% & \textbf{14.20\%} & 5.50\% & 13.00\% & 10.90\% \\
		\hline
	\end{tabular}
}
	\caption{Type-1 error (\%) for different target allocations with fixed targeting design (Distance), 1000 simulations, n = 200}
	\label{tab:fixed-targeting}
\end{table}

Figure~\ref{fig:fixed-tgting-power} reports the empirical power as a function of the number of patients $n$ for the configuration
	\(\Theta=(0.6,0.65,0.7,0.75)\) using the Distance targeting design with forced exploration. The results are based on \(n_{\mathrm{sim}}=1000\) independent simulation runs and the shaded bands correspond to 95\% confidence intervals. As expected, the power increases with the sample size for all designs. While the Uniform, Neyman and RSIHR allocations exhibit very similar power profiles,
	 the Tymofyeyev allocation yields significantly lower power, on all the range of sample sizes.  

	 \begin{figure}[h]
	\centering
	\includegraphics[width=0.7\linewidth]{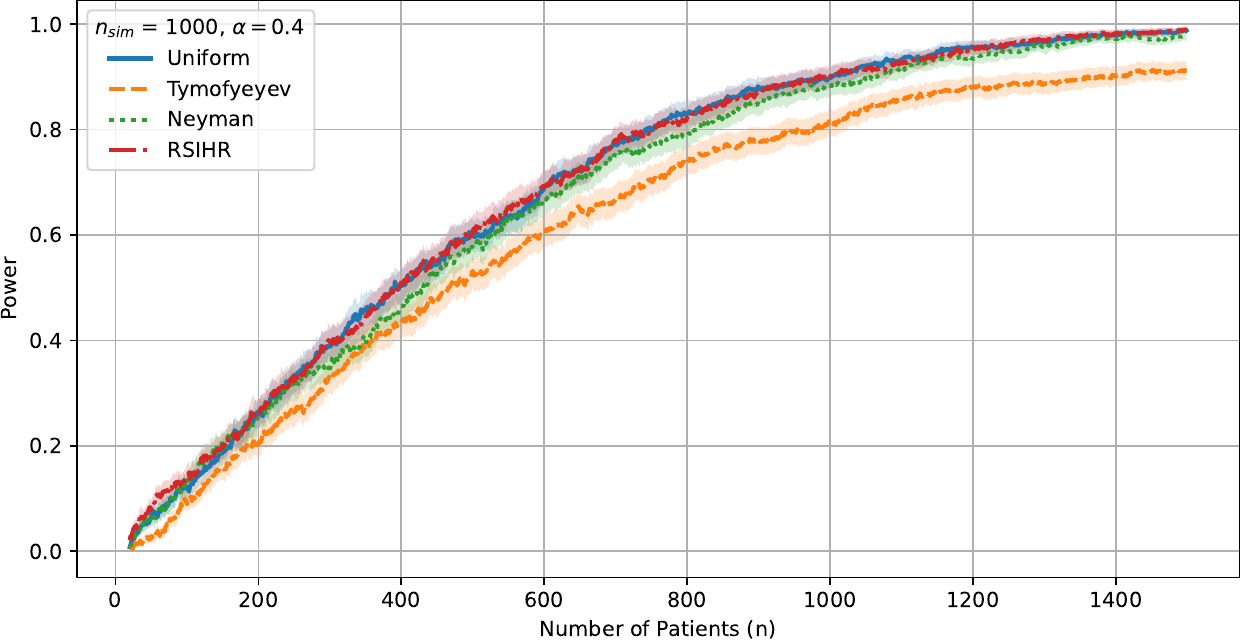}
	\caption{Power profile for fixed Distance Targeting}
	\label{fig:fixed-tgting-power}
\end{figure}

	This last observation is perhaps surprising as we recall that the Tymofyeyev target allocation is the one maximizing the power of this homogeneity test. However, for moderate sample sizes, we observe that a RAR procedure converging to this optimal allocation (even with minimal asymptotic variance) does not necessarily inherit this maximal power.

	\section{Conclusion}\label{sec:conclusion}
	In this paper, we introduced a unified theoretical framework for response-adaptive targeting in multi-treatment experiments. Rather than focusing on a specific allocation rule, we defined a broad class of adaptive designs, the $\alpha$RTS family, encompassing several existing procedures while allowing for new algorithmic variations.
	
	We established that all designs within this class share the same fundamental asymptotic properties, including strong consistency, asymptotic normality of both estimators and allocation proportions, and asymptotic efficiency. A key insight of our analysis is that these properties are structural and do not depend on the specific form of the allocation rule within the class. To address sparse target regimes, we proposed a forced-exploration extension and showed that it preserves the asymptotic guarantees while ensuring sufficient sampling of all treatments. Our simulation results illustrate that, while different designs may exhibit distinct finite-sample behaviors, they remain asymptotically equivalent, and highlight the importance of forced exploration in sparse settings.
	
	Overall, our work provides a flexible and robust foundation for designing and analyzing response-adaptive procedures in multi-arm trials.
	
	\subsection*{Acknowledgments}
		The PhD of Redouane Yagouti is founded by an Inria Action Exploration grant called BETA-3$K$. The authors acknowledge Sof\'ia Villar and Lukas Pin for interesting discussions on Response Adaptive Randomization.

	\bibliographystyle{plainnat}
	\bibliography{references}

@article{hu2004asymptotic,
	title={Asymptotic properties of doubly adaptive biased coin designs for multitreatment clinical trials},
	author={Hu, Feifang and Zhang, Li-Xin},
	journal={The Annals of Statistics},
	volume={32},
	number={1},
	pages={268--301},
	year={2004},
	publisher={Institute of Mathematical Statistics}
}

@article{HuRosenberger03VariabilityPower,
title = {Optimality, Variability, Power
Evaluating Response-Adaptive Randomization Procedures for Treatment Comparisons},
author = {Feifang Hu and William Rosenberger},
journal = {Journal of the American Statistical Society},
volume = {98 (463)},
pages = {671-678}
}

@article{alkhnefr2025efficient,
	title={Efficient randomized adaptive designs for multi-arm clinical trials},
	author={Alkhnefr, Norah and Hu, Feifang and Zhai, Guannan},
	journal={Statistical Methods in Medical Research},
	volume={34},
	number={9},
	pages={1886--1898},
	year={2025},
	publisher={SAGE Publications Sage UK: London, England}
}

@article{Efron71BCD,
author = {Efron, Bradley},
title = {Forcing a sequential experiment to be balanced},
year = {1971},
journal = {Biometrika},
volume = {58},
pages = {403–417}
}

@article{hu2009efficient,
	title={Efficient randomized-adaptive designs},
	author={Hu, Feifang and Zhang, Li-Xin and He, Xuming},
	journal={The Annals of Statistics},
	pages={2543--2560},
	year={2009},
	publisher={JSTOR}
}

@article{Villar23Review,
author = {David S. Robertson  and  Kim May Lee  and Boryana C. López-Kolkovska  and Sofía S. Villar},
title = {Response-adaptive randomization in clinical trials: from myths to practical considerations},
journal = {Statistical Science},
volume = {38(2)},
year = {2023}
}

@Article{Thompson33,
  Title                    = {{On the likelihood that one unknown probability exceeds another in view of the evidence of two samples}},
  Author                   = {Thompson, W.R.},
  Journal                  = {Biometrika},
  Year                     = {1933},
  Pages                    = {285--294},
  Volume                   = {25}
}

@Article{Robbins52Freq,
  Title                    = {{Some aspects of the sequential design of experiments}},
  Author                   = {Robbins, H.},
  Journal                  = {Bulletin of the American Mathematical Society},
  Year                     = {1952},
  Pages                    = {527--535},
  Volume                   = {58(5)}
}

@ARTICLE{Thall07,
	TITLE = {Practical Bayesian Adaptive Randomization in Clinical Trials},
	AUTHOR = {P.F. Thall and J.K. Wathen},
	JOURNAL = {European Journal on Cancer},
	YEAR = {2007},
	VOLUME = {43},
	PAGES = {859-866}
}

@article{Rosenberger01Opt,
title = {Optimal Adaptive Designs for Binary Response Trials},
author = {Rosenberger, W.F. and Stallard, N. and Ivanova, A. and Harper, C.N. and Ricks, M.L.},
journal = {Biometrics},
volume = {57},
pages = {909–913},
year = {2001}}

@article{WeiDurham78RPW,
author={Wei,  LJ and  Durham S.},
title = {The Randomized Play-the-winner Rule in Medical Trials},
journal = {Journal of Medical Statistics Association},
year = {1978},
volume = {73},
pages = {840–843}
}

@article{Bartlett85ECMO,
author = {Bartlett, R. and Roloff D. and Cornell R. and Andrews A. and Dillon P. and  Zwischenberger J.},
title = {Extracorporeal Circulation in Neonatal Respiratory Failure: A Prospective Randomized Study},
journal = {Pediatrics Journal},
year = {1985},
volume = {76 (4)},
pages = {479–487}}

@book{BanditBook,
author = {Lattimore, Tor and Szepesvari, Csaba},
publisher = {Cambridge University Press},
title = {{Bandit Algorithms}},
year = {2019}
}

@article{hu2006asymptotically,
	title={Asymptotically best response-adaptive randomization procedures},
	author={Hu, Feifang and Rosenberger, William F and Zhang, Li-Xin},
	journal={Journal of Statistical Planning and Inference},
	volume={136},
	number={6},
	pages={1911--1922},
	year={2006},
	publisher={Elsevier}
}

@book{hu2006theory,
	title={The theory of response-adaptive randomization in clinical trials},
	author={Hu, Feifang and Rosenberger, William F},
	year={2006},
	publisher={John Wiley \& Sons}
}

@article{Doob1936,
	title = {Note on Probability},
	author = {Doob, J. L.}, 
	year = {1936}, 
	journal = {Ann. Math},
	volume ={37},
	pages = {363-367}
}

@article{Chow1967,
	title = {On a Strong Law of Large Numbers for Martingales},
	author = {Chow, Y. S.}, 
	year = {1967}, 
	journal = {Ann. Math. Statist.},
	volume ={328},
	pages = {610-611}
}

@article{Stout1970,
	title = {A Martingale Analogue of Kolmogorov's Law
	of the Iterated Logarithm},
	author = {Stout, W. F.}, 
	year = {1970}, 
	journal = {Z. Wahrscheinlichkeitstheorie verw. Gebiete},
	volume ={15},
	pages = {279-290}
}

@article{melfi1998variablility,
	title={Variablility in adaptive designs for estimation of success probabilities},
	author={Melfi, Vincent and Page, Connie},
	journal={Lecture Notes-Monograph Series},
	pages={106--114},
	year={1998},
	publisher={JSTOR}
}

@misc{jf-le-gall,
	author       = {Jean-Fran\c{c}ois Le Gall},
	title        = {Int\'egration, Probabilit\'es et Processus Al\'eatoires},
	year         = {2006},
	howpublished = {Cours de Master 2, Universit\'e Paris-Sud},
	note         = {Lecture notes}
}

@Article{KLUCBJournal,
  Title                    = {{{K}ullback-{L}eibler upper confidence bounds for optimal sequential allocation}},
  Author                   = {Capp{\'e}, O. and Garivier, A. and Maillard, O-A. and Munos, R. and Stoltz, G.},
  Journal                  = {Annals of Statistics},
  Year                     = {2013},
  Pages                    = {1516--1541},
  Volume                   = {41(3)}
}

@book{marc-yor,
	author    = {Daniel Revuz and Marc Yor},
	title     = {Continuous Martingales and Brownian Motion},
	publisher = {Springer},
	year      = {1999},
	edition   = {3rd},
	series    = {Grundlehren der mathematischen Wissenschaften},
	volume    = {293},
	address   = {Berlin}
}

@inproceedings{garivier2016optimal,
	title={Optimal best arm identification with fixed confidence},
	author={Garivier, Aur{\'e}lien and Kaufmann, Emilie},
	booktitle={Conference on Learning Theory},
	pages={998--1027},
	year={2016},
	organization={PMLR}
}

@article{tymofyeyev2007implementing,
	title={Implementing optimal allocation in sequential binary response experiments},
	author={Tymofyeyev, Yevgen and Rosenberger, William F and Hu, Feifang},
	journal={Journal of the American Statistical Association},
	volume={102},
	number={477},
	pages={224--234},
	year={2007},
	publisher={Taylor \& Francis}
}

@incollection{pin2024response,
	title={Response-adaptive randomization designs based on optimal allocation proportions},
	author={Pin, Lukas and Villar, Sofia S and Rosenberger, William F},
	booktitle={Biostatistics in Biopharmaceutical Research and Development: Clinical Trial Design, Volume 1},
	pages={313--339},
	year={2024},
	publisher={Springer}
}

@article{neyman1934representative,
	author  = {Neyman, Jerzy},
	title   = {On the two different aspects of the representative method: The method of stratified sampling and the method of purposive selection. },
	journal = {Journal of the Royal Statistical Society},
	volume   = {97},
	number   = {4},
	pages    = {558--606},
	year     = {1934}
}

@book{Cochran1977,
	author    = {Cochran, William G.},
	title     = {Sampling Techniques},
	year      = {1977},
	edition   = {3rd},
	publisher = {John Wiley \& Sons},
	address   = {New York}
}
	
	\newpage
	\begin{appendix}
		\section{Supporting Proofs for $\alpha$RTS}\label{appnA}
		
We fix a design that belongs to the \(\alpha\)RTS family for some \(0\leq\alpha <1\) and introduce the following notation for \(k \in [K]\):
\begin{itemize}
	\item $p_{m,k} = \mathbb{E}[X_{m,k} \mid \mathcal{F}_{m-1}]$	
	\item $\Delta M_{m,k} = X_{m,k} - \mathbb{E}[X_{m,k} \mid \mathcal{F}_{m-1}] = X_{m,k} - p_{m,k}$
	\item $Q_{n,k} = \sum_{m=1}^{n}X_{m,k}(\xi_{m,k} - \theta_{k})$ and denote \(\mathbf{Q_n} := (\mathbf{Q_{n,1}}, \dots, \mathbf{Q_{n,K}})\)	
\end{itemize}
and
	\[U_{n,k} = \sum_{m=1}^{n-1} \alpha \hat{\rho}_{m,k} + M_{n,k} - n \hat{\rho}_{n,k}\;.\]
We define the last hitting time $\ell_{n,k}$ as
\begin{equation}\label{eq:stopping-time-def}
	\ell_{n,k} := \max \left\{\, m \leq n ; \, \frac{N_{m,k}}{m} \leq \hat{\rho}_{m,k}  \;.\right\}
\end{equation}
We first prove a key lemma that is central to establish our main theorems.

\begin{lemma}\label{lem:preliminary}
	For any \(k \in [K]\) we have the following results:
	\begin{enumerate}[label=(\roman*)]
		\item \(\left(\ell_{n,k}\right)_{n \geq 1}\) is a non-decreasing sequence and \(\forall n \in \mathbb{N}: \,\ell_{n,k} \leq n\)
		\item for all $k\in [K]$, $n \in \mathbb{N}$,\begin{equation}\label{eq:prelim-4}
			N_{n,k} - n\hat{\rho}_{n,k} \leq 1+N_{\ell_{n,k},k} - \ell_{n,k} \hat{\rho}_{\ell_{n,k},k} + U_{n,k} - U_{\ell_{n,k},k}	
		\end{equation}
		\item for all $k\in [K]$, $n \in \mathbb{N}$,
		\begin{equation}\label{eq:condition}
			N_{\ell_{n,k},k} - \ell_{n,k} \hat{\rho}_{\ell_{n,k},k} \leq o(\sqrt{n}), \ a.s.
			\quad 
		\end{equation}
	\end{enumerate}
\end{lemma}
\begin{proof}
	\begin{enumerate}[label=(\roman*)]
		\item Obvious from definition \eqref{eq:stopping-time-def}.
		\item By definition of \(U_{n,k}\) we have \(\Delta U_{m,k} = \alpha \hat{\rho}_{m-1,k} - p_{m,k} 
		+ \Delta \Big( N_{m,k} - m \hat{\rho}_{m,k} \Big)\).
%
		Then 
		\begin{align*}
			U_{n,k} - U_{\ell_{n,k},k}
			&= \sum_{m=\ell_{n,k}+1}^n \Big( \alpha \hat{\rho}_{m-1,k} - p_{m,k} \Big)
			+ N_{n,k} - n \hat{\rho}_{n,k}
			- \Big( N_{\ell_{n,k},k} - \ell_{n,k} \hat{\rho}_{\ell_{n,k},k} \Big)
		\end{align*}
		So 
		\begin{align*}
			N_{n,k} - n \hat{\rho}_{n,k}
			&= N_{\ell_{n,k},k} - \ell_{n,k} \hat{\rho}_{\ell_{n,k},k}
			+ p_{\ell_{n,k}+1,k} - \alpha \hat{\rho}_{\ell_{n,k},k} \\&\,\,+ \sum_{m=\ell_{n,k}+2}^n \underbrace{\Big( p_{m,k} - \alpha \hat{\rho}_{m-1,k} \Big)}_{\leq 0 \quad}
			+ U_{n,k} - U_{\ell_{n,k}, k}\\					
			&\leq p_{\ell_{n,k}+1,k} + N_{\ell_{n,k},k} - \ell_{n,k} \hat{\rho}_{\ell_{n,k},k}
			+ U_{n,k} - U_{\ell_{n,k}, k}
		\end{align*}
		So we conclude 	
		\[N_{n,k} - n\hat{\rho}_{n,k} \leq 1 + N_{\ell_{n,k},k} - \ell_{n,k} \hat{\rho}_{\ell_{n,k},k} + U_{n,k} - U_{\ell_{n,k},k}\;.\]
		
		\item It is easy to check that if $\ell_{n,k} < K m_0+1$ then \(N_{\ell_{n,k},k} - \ell_{n,k} \hat{\rho}_{\ell_{n,k},k} \leq K m_0\),
		and if \(\ell_{n,k} \geq Km_0 + 1\) then \(N_{\ell_{n,k},k} - \ell_{n,k} \hat{\rho}_{\ell_{n,k},k} \leq 0\).
		So by combining the two inequalities we obtain that \(N_{\ell_{n,k},k} - \ell_{n,k} \hat{\rho}_{\ell_{n,k},k} \leq K m_0 = o(\sqrt{n})\)\;.
	\end{enumerate}
\end{proof}

\subsection{Proof of Theorem \ref{thm:LLN}}\label{proof:thm1}

We now explain how Lemma~\ref{lem:preliminary} leads to Theorem~\ref{thm:LLN}.

\begin{proof}[Proof of Theorem \ref{thm:LLN}]	
The first step of the proof, as in the work of \cite{hu2004asymptotic}, is to establish that $\hat{\rho}_{n}$ has a finite (non-sparse) limit.
	
	For \(k \in [K]\), if \(N_{n,k} \to +\infty\), then using  Condition \hyperlink{CA}{\textbf{A}} and the law of large numbers 
	we get \(\hat\theta_{n,k} \to \theta_k\). Otherwise if \(\sup_{n \geq 1}N_{n,k} < +\infty\) then \(\left(\hat\theta_{n,k}\right)_{n \geq 1}\) takes finite values and also has a limit, that belong to $V_k$ (which is defined in our statement of Condition  \hyperlink{CB}{\textbf{B}}). Hence $\hat\theta_{n}$ converges to some element $z$ in $I_1\times \dots \times I_K$.  By continuity of $\rho$ and from Condition \hyperlink{CB}{\textbf{B}} we obtain
	\begin{equation}\label{eq:conv-rho}
		\exists u \in (0,1)^K: \quad \hat{\rho}_n \to u:= \rho(z)\;.
	\end{equation}
	
	As \(M_n\) is a martingale, we have by LLN for martingales (see \cite{Chow1967}) that $M_{n,k} = o(n)$ a.s. so
	\begin{align*}
		\frac{U_{n,k}}{n} &= \frac{1}{n} \sum_{m=0}^{n-1} \alpha \hat{\rho}_{m,k} - \hat{\rho}_{n,k} + o(1) \\
		\implies \lim_{n \to \infty} \frac{U_{n,k}}{n} &= - (1 - \alpha) u_k \quad a.s. \quad\text{; combining \eqref{eq:conv-rho} and the Cesaro lemma}
	\end{align*}
	Using Lemma \ref{lem:convergence} and Inequality \eqref{eq:condition} from Lemma \ref{lem:preliminary} we obtain that
	\[
	\left(\frac{1}{n}+\frac{N_{\ell_{n,k},k} - \ell_{n,k} \hat{\rho}_{\ell_{n,k},k}}{n} + \frac{U_{n,k} - U_{\ell_{n,k},k}}{n} \right)^+ \to 0 \quad a.s.
	\]
	So by injecting the previous result in Inequality \eqref{eq:prelim-4} from Lemma \ref{lem:preliminary}, we obtain
	\[
	\forall k \in [K]: \, \left(\frac{N_{n,k}}{n} - \hat{\rho}_{n,k}\right)^+ \to 0 \quad a.s.
	\]
	and we have
	\begin{align*}
		\left(\hat{\rho}_{n,k} - \frac{N_{n,k}}{n}\right)^+ &= \left(1 - \frac{N_{n,k}}{n} - (1 - \hat{\rho}_{n,k})\right)^+ \\
		&= \left(\sum_{j=1; j \neq k}^K \frac{N_{n,j}}{n} - \hat{\rho}_{n,j}\right)^+ \\
		&\leq \sum_{j=1; j \neq k}^K \left(\frac{N_{n,j}}{n} - \hat{\rho}_{n,j}\right)^+  \to 0
	\end{align*}
	
	So we obtain that \(\frac{N_{n,k}}{n} - \hat{\rho}_{n,k} \to 0\),  and then \(\lim_{n \to \infty} \frac{N_{n,k}}{n} = \lim\limits_{n \to \infty} \hat{\rho}_{n,k} = u_k \in (0,1)\)
	
	A straightforward consequence is that $\forall k \in [K], N_{n,k} \to \infty$ a.s., hence $\hat{\Theta}_n \to \Theta$ and by continuity of \(\rho\) we obtain that \(\hat{\rho}_n = \rho(\hat{\Theta}_n) \to \rho(\Theta) = v \in (0,1)^K\), hence
	\[
	\lim_{n \to \infty} \frac{N_{n,k}}{n} = \lim_{n \to \infty} \hat{\rho}_{n,k} = v_k
	\]
	Using further Lemma A.4 \cite{hu2004asymptotic} and the Law of Iterated Logarithm we get
	\[
	\hat{\Theta}_n - \Theta = O\left(\sqrt{\frac{\log\!\log n}{n}}\right) \quad \text{a.s.}
	\]
	
	In addition, the second order Taylor expansion of \(\rho\)  gives that
	\begin{align*}
		\hat{\rho}_n - v  &= (\hat{\Theta}_n-\Theta)\frac{\partial \rho}{\partial y}\Big|_{\Theta} + O(\|\hat{\Theta}_n-\Theta\|^2) \\
		&= O\left(\sqrt{\frac{\log\!\log n}{n}}\right) \quad \text{a.s.}
	\end{align*}
	
\end{proof}

\subsection{Proof of Theorem~\ref{thm:normality-gen-erade}}\label{proof:thm2}To prove Theorem~\ref{thm:normality-gen-erade}, we need a more precise characterization of the behavior of $U_{n,k} - U_{\ell_{n,k},n}$, established in a series of lemma. We observe that the proofs of these lemmas only use the conclusions of Lemma~\ref{lem:preliminary}.

\begin{lemma}\label{lem:diff-U-approx}
	Let $(\ell_n)_{n \geq 1}$ a positive increasing sequence such that $\forall n \in \mathbb{N}: \ell_n \leq n$. Under conditions \hyperlink{CA}{\textbf{A}} and \hyperlink{CB}{\textbf{B}} we have the following
	\begin{enumerate}[label=(\roman*)]
		\item For all \(k \in [K]\) 
		\begin{equation}\label{eq:prelim-2}
			U_{n,k} - U_{\ell_{n},k} = (n - \ell_{n})[-(1 - \alpha)v_k + o(1)] + M_{n,k} - M_{\ell_{n},k} + \ell_{n} (\hat{\rho}_{\ell_{n},k} - \hat{\rho}_{n,k})
		\end{equation}
		\item There exists a constant \(C \geq 0\) such that
		\begin{equation}\label{eq:prelim-22}
			|\ell_n(\hat{\rho}_n - \hat{\rho}_{\ell_n})| \le o(1)(n-\ell_n) + C\|Q_n - Q_{\ell_n}\|  + o_p(\sqrt{n})
		\end{equation}		
	\end{enumerate}
\end{lemma}
\begin{proof}
	\begin{enumerate}[label=(\roman*)]
		\item \begin{align*}
			&U_{n,k} - U_{\ell_{n},k} = \sum_{m=\ell_{n}}^{n-1} \alpha \hat{\rho}_{m,k} - n \hat{\rho}_{n,k} + \ell_{n} \hat{\rho}_{\ell_{n},k} + M_{n,k} - M_{\ell_{n,k},k} \\
			=& \sum_{m=\ell_{n}}^{n-1} \alpha (\hat{\rho}_{m,k} - v_k) + \alpha(n - \ell_{n}) v_k - (n - \ell_{n}) \hat{\rho}_{n,k} + \ell_{n} (\hat{\rho}_{\ell_{n},k} - \hat{\rho}_{n,k}) \\ 
			&\quad + M_{n,k} - M_{\ell_{n},k} \\
			=& -(1 - \alpha)(n - \ell_{n}) v_k + \sum_{m=\ell_{n}}^{n-1} \alpha (\hat{\rho}_{m,k} - v_k) - (n - \ell_{n})(\hat{\rho}_{n,k} - v_k) + M_{n,k} - M_{\ell_{n},k} \\
			&\quad + \ell_{n} (\hat{\rho}_{\ell_{n},k} - \hat{\rho}_{n,k}) \\ 
			=& (n-\ell_{n})\left[-(1 - \alpha)v_k - (\hat{\rho}_{n,k} - v_k) + \frac{\alpha}{n - \ell_{n}}\sum_{m=\ell_{n}}^{n-1} (\hat{\rho}_{m,k} - v_k)\right]+ M_{n,k} - M_{\ell_{n},k} \\
			&\quad + \ell_{n} (\hat{\rho}_{\ell_{n},k} - \hat{\rho}_{n,k}) \\
		\end{align*}
		Theorem \ref{thm:LLN} applies here and gives that \(\lim_{n \to +\infty} \hat{\rho}_n = v\). If \(\left(\ell_{n}\right)_{n \geq 1}\) is bounded then it is stationary, so the Cesaro yields \(\frac{1}{n - \ell_{n}}\sum_{m=\ell_{n}}^{n-1} (\hat{\rho}_{m,k} - v_k) = o(1)\). Otherwise if \(\ell_{n} \to +\infty\), we have
		\begin{align*}
			\frac{1}{n - \ell_{n}}\left|\sum_{m=\ell_{n}}^{n-1} \hat{\rho}_{m,k} - v_k\right| &\leq \frac{1}{n - \ell_{n}}\sum_{m=\ell_{n}}^{n-1} |\hat{\rho}_{m,k} - v_k| \\
			&\leq \sup_{\ell_{n} \leq m \leq n}|\hat{\rho}_{m,k} - v_k| \to 0\;.
		\end{align*}		
		In both cases we obtain
		\begin{equation*}
			U_{n,k} - U_{\ell_{n},k} = (n - \ell_{n})[-(1 - \alpha)v_k + o(1)] + M_{n,k} - M_{\ell_{n},k} + \ell_{n} (\hat{\rho}_{\ell_{n},k} - \hat{\rho}_{n,k})
		\end{equation*}
		\item 
		Let \( D_{n,k} := \sum_{j=1}^{n} X_{j,k}(\xi_{j,k}-\theta_k) \), then using Lemma \ref{lem:chebychev-mart} we have for \( M>0 \)		
		\begin{align*}
			\sup_{n\geq 1} \mathbb{P}(D_{n,k} > M\sqrt{n})
			&\leq \sup_{n\geq 1}
			\mathbb{P}\!\left(\max_{1\leq i\leq n} D_{i,k} \geq M\sqrt{n}\right) \\
			&\leq \sup_{n\geq 1}
			\frac{\mathbb{V}ar\left[D_{n,k}\right]}
			{\mathbb{V}ar\left[D_{n,k}\right] + M^2 n} \\
			&\leq
			\frac{\mathbb{V}ar[\xi_{1,k}]}
			{\mathbb{V}ar[\xi_{1,k}] + M^2}.
		\end{align*}
		
		It follows that \( D_{n,k} = O_p(\sqrt{n}) \). Moreover, we have \(\frac{\sqrt{n}}{N_{n,k}}  =\frac{1}{\sqrt{n}}\frac{n}{N_{n,k}} = \frac{1}{\sqrt{n}} O_p(1) = O_p\left(\frac{1}{\sqrt{n}}\right)\) and \(N_{n,k}^{-1/2} = \frac{1}{\sqrt{n}}\sqrt{n}N_{n,k}^{-1/2} = O_p\left(\frac{1}{\sqrt{n}}\right)\). Using the Bahadur representation of the estimator in \eqref{eq:bahadur-estimator} we obtain
			\begin{align*}
				\hat{\theta}_{n,k} - \theta_k
				&=
				\frac{D_{n,k}}{N_{n,k}}
				+ o_p\!\left(N_{n,k}^{-1/2}\right) \\
				&=
				O_p\!\left(\frac{\sqrt{n}}{N_{n,k}}\right)
				+ o_p\!\left(N_{n,k}^{-1/2}\right) \\
				&=
				O_p\!\left(\frac{1}{\sqrt{n}}\right) + o_p\!\left(\frac{1}{\sqrt{n}}\right)\\
				&=
				O_p\!\left(\frac{1}{\sqrt{n}}\right).
		\end{align*}
		
		So the second order Taylor expansion of \(\rho\) gives
		\begin{align}
			\hat{\rho}_n - v  &= (\hat{\Theta}_n-\Theta)\frac{\partial \rho}{\partial y}\Big|_{\Theta} + O(\|\hat{\Theta}_n-\Theta\|^2) \\
			&= (\hat{\Theta}_n-\Theta)\frac{\partial \rho}{\partial y}\Big|_{\Theta} + O_p\left(\frac{1}{n}\right) \label{eq:approx-rho-Op}
		\end{align}
		Similarly we obtain \(		\hat{\rho}_{\ell_n} - v = (\hat{\Theta}_{\ell_n}-\Theta)\frac{\partial \rho}{\partial y}\Big|_{\Theta} + O_p\left(\frac{1}{\ell_n}\right)
		\).	Combining the two previous equations gives
		\begin{eqnarray*}
			\ell_n(\hat{\rho}_{\ell_n}-\hat{\rho}_n) = \ell_n(\hat{\Theta}_{\ell_n}-\hat{\Theta}_n)\frac{\partial \rho}{\partial y}\Big|_{\Theta} + \left(O_p\left(\frac{\ell_n}{n}\right)+O_p(1)\right) \\
			\implies |\ell_n(\hat{\rho}_{\ell_n}-\hat{\rho}_n)| \le C \ell_n \|\hat{\Theta}_{\ell_n}-\hat{\Theta}_n\| + O_p(1)
		\end{eqnarray*}
		And we have	
		\begin{align*}
			\hat{\theta}_{n,k}-\hat{\theta}_{\ell_n,k} &= \frac{1}{N_{n,k}}\sum_{j=1}^{n}X_{j,k}(\xi_{j,k}-\theta_k)-\frac{1}{N_{\ell_n,k}}\sum_{j=1}^{\ell_n}X_{j,k}(\xi_{j,k}-\theta_k)\\
			&= \frac{1}{N_{n,k}}\sum_{j=\ell_n+1}^{n} X_{j,k}(\xi_{j,k}-\theta_k) + \left(\frac{1}{N_{n,k}} - \frac{1}{N_{\ell_n,k}}\right)\sum_{j=1}^{\ell_n}X_{j,k}(\xi_{j,k}-\theta_k)\\[8pt]
			&= \frac{1}{N_{n,k}}\sum_{j=\ell_n+1}^{n}X_{j,k}(\xi_{j,k}-\theta_k) + \frac{N_{\ell_n,k}-N_{n,k}}{N_{n,k}}\cdot\frac{1}{N_{\ell_n,k}}\sum_{j=1}^{\ell_n}X_{j,k}(\xi_{j,k}-\theta_k)\\[8pt]
			&\quad +\, o\left(N_{n,k}^{-1/2}\right) + o\left(N_{\ell_n,k}^{-1/2}\right)
		\end{align*}
		So
		\begin{align*}
			|\ell_n(\hat{\rho}_{\ell_n}-\hat{\rho}_n)| &\le C\sum_{k=1}^{K}\frac{\ell_n}{N_{n,k}}\left\|\sum_{j=\ell_n+1}^{n}X_{j,k}(\xi_{j,k}-\theta_k)\right\|\\
			&\quad + C\sum_{k=1}^{K}\frac{\ell_n|N_{\ell_n,k}-N_{n,k}|}{N_{n,k}}\left\|\frac{1}{N_{\ell_n,k}}\sum_{j=1}^{\ell_n}X_{j,k}(\xi_{j,k}-\theta_k)\right\|\\
			&\quad + \ell_n \cdot o\left(N_{n,k}^{-1/2}\right) + \ell_n \cdot o\left(N_{\ell_n,k}^{-1/2}\right)
		\end{align*}
		And for some two sequences $a_n$ and $b_n$ converging to $0$ as $n \to + \infty$
		\begin{align*}
			\ell_n\cdot o\left(N_{n,k}^{-1/2}\right)+\ell_n\cdot o\left(N_{\ell_n,k}^{-1/2}\right)
			&= \ell_n\,N_{n,k}^{-1/2}a_n + \ell_n \,N_{\ell_n,k}^{-1/2}b_n\\
			&= \sqrt{n}\left(\frac{\ell_n}{n}\left(\frac{N_{n,k}}{n}\right)^{-\frac{1}{2}}a_{n} + \sqrt{\frac{\ell_n}{n}}\left(\frac{N_{\ell_n,k}}{\ell_n}\right)^{-\frac{1}{2}}b_{n}\right)\\
			&= o_p(\sqrt{n})
		\end{align*}	
		as \quad \(\lim\limits_{n\to+\infty} \frac{\ell_n}{n}\left(\frac{N_{n,k}}{n}\right)^{-1/2}a_n + \sqrt{\frac{\ell_n}{n}}\left(\frac{N_{\ell_n,k}}{\ell_n}\right)^{-1/2}b_n=0\)\;. Since \(\ell_n \le n\) and \(\frac{n}{N_{n,k}}\to v_k^{-1}\) we have
		\begin{align*}
			C\sum_{k=1}^{K}\frac{\ell_n}{N_{n,k}}\left\|\sum_{m=\ell_n+1}^{n}X_{m,k}(\xi_{m,k}-\theta_k)\right\|
			&=C\sum_{k=1}^{K}\frac{\ell_n}{n}\frac{n}{N_{n,k}}\|Q_{n,k}-Q_{\ell_n,k}\|\\
			&\le C\sum_{k=1}^{K}\|Q_{n,k}-Q_{\ell_n,k}\|\\
			&= C\|Q_n-Q_{\ell_n}\|
		\end{align*}	
		On the second part of the inequality we use that \(N_{n,k} - N_{\ell_{n}, k} \leq n - \ell_n\):
		\begin{align*}
			&C\sum_{k=1}^{K}\frac{\ell_n|N_{\ell_n,k}-N_{n,k}|}{N_{n,k}}\left\|\frac{1}{N_{\ell_n,k}}\sum_{m=1}^{\ell_n}X_{m,k}(\xi_{m,k}-\theta_k)\right\|\\&= C\sum_{k=1}^{K}\frac{\ell_n|N_{\ell_n,k}-N_{n,k}|}{N_{\ell_n,k}N_{n,k}}\left\|\sum_{m=1}^{\ell_n}X_{m,k}(\xi_{m,k}-\theta_k)\right\|\\
			& \leq C(n-\ell_n) \sum_{k=1}^{K}\frac{\ell_n}{N_{\ell_n,k}}\frac{n}{N_{n,k}} \left\|\frac{1}{n}\sum_{m=1}^{\ell_n}X_{m,k}(\xi_{m,k}-\theta_k)\right\|
		\end{align*}
		If \((\ell_{n})_{n\geq1}\) is bounded, then \(\frac{1}{n}\left\|\sum_{m=1}^{\ell_n}X_{m,k}(\xi_{m,k}-\theta_k)\right\| = o(1)\). Otherwise if \(\ell_n \to +\infty\) then we can use the argument of \cite{Doob1936} and conclude that
		\[\frac{1}{\ell_n}\left\|\sum_{m=1}^{\ell_n}X_{m,k}(\xi_{m,k}-\theta_k)\right\| = o(1)\]
		Combined with the fact that \(N_{n,k} \sim v_k.n\), we obtain that
		\begin{align*}
			\frac{\ell_n}{N_{\ell_n,k}}\frac{n}{N_{n,k}}\left\|\frac{1}{n}\sum_{m=1}^{\ell_n}X_{m,k}(\xi_{m,k}-\theta_k)\right\| = o(1)
		\end{align*}
		So
		\begin{equation*}
			C\sum_{k=1}^{K}\frac{\ell_n|N_{\ell_n,k}-N_{n,k}|}{N_{n,k}}\left\|\frac{1}{N_{\ell_n,k}}\sum_{m=1}^{\ell_n}X_{m,k}(\xi_{m,k}-\theta_k)\right\| \leq (n-\ell_n) o(1)
		\end{equation*}
		Combining the previous results we obtain
		\[
		|\ell_n(\hat{\rho}_n - \hat{\rho}_{\ell_n})| \le o(1)(n-\ell_n) + C\|Q_n - Q_{\ell_n}\| + o_p(\sqrt{n})
		\]
	\end{enumerate}
\end{proof}

\begin{lemma}\label{lem:U-ineq}
	Under conditions \hyperlink{CA}{\textbf{A}}--\hyperlink{CB}{\textbf{B}} , we have for every \(k \in [K]\) as \(n \to +\infty\)	
	\begin{align}
		U_{n,k} - U_{\ell_{n,k},k} &\leq o_p(\sqrt{n})\label{eq:ineq-U-1}\\
		U_{n,k} - U_{\ell_{n,k},k} &\leq O(\sqrt{n \log \log(n)}) \quad \text{a.s.}\label{eq:ineq-U-2}
	\end{align}
\end{lemma}	
\begin{proof}
	Let \((L_n)_{n \geq 1}\) a positive sequence such that $L_n \to \infty$ and $L_n = o(n)$. From Lemma \ref{lem:mart-1} we have
	\begin{equation*}
		\mathbb{E}\left[\max_{m \leq L_n} \|Q_n - Q_{n-m}\| \right] \leq 2\sqrt{L_n C_0}
		\quad \text{and} \quad \mathbb{E}\left[\max_{L_n \leq m < n} \frac{\|Q_n - Q_{n-m}\|}{m} \right] \leq 3\sqrt{\frac{C_0}{L_n}}  
	\end{equation*}
	Applying Lemma \ref{lem:expectation-to-O} we obtain,
	\begin{align*}
		\max_{L_n \leq m < n} \frac{\|Q_m - Q_{n-m}\|}{m} &= O_p\left(\frac{1}{\sqrt{L_n}}\right) = o_p(1)\\ 
		\max_{L_n \leq m < n} \|Q_m - Q_{n-m}\| &= O_p(\sqrt{L_n}) \leq o_p(\sqrt{n})
	\end{align*}
	From Lemma \ref{lem:mart-2}, we have that for all \(L < n\):
	\begin{equation}\label{eq:ineq-proof-1}
		\|Q_n - Q_{\ell_{n,k}}\| \leq (n - \ell_{n,k}) \max_{L \leq m < n} \frac{\|Q_n - Q_{n-m}\|}{m} + \max_{m < L} \|Q_n - Q_{n-m}\|
	\end{equation}
	So using \eqref{eq:ineq-proof-1}, and applying the same steps on the martingale \(\{M_n\}_{n \geq 1}\), we obtain
	\begin{equation}\label{eq:ineq-M-U}
		\|Q_n - Q_{\ell_{n,k}}\| \leq o_p(1)(n - \ell_{n,k}) + o_p(\sqrt{n}), \quad |M_{n,k} - M_{\ell_{n,k},k}| \leq o_p(1)(n - \ell_{n,k}) + o_p(\sqrt{n})
	\end{equation}	
	From inequality \eqref{eq:prelim-22} of Lemma \ref{lem:diff-U-approx} and \eqref{eq:ineq-M-U}
	\begin{align}
		|\ell_{n,k} (\hat{\rho}_{\ell_{n,k},k} - \hat{\rho}_{n,k})| &\leq o(1)(n - \ell_{n,k}) + C \|Q_n - Q_{\ell_{n,k}}\| + o_p(\sqrt{n})\notag\\
		&\leq o_p(1)(n - \ell_{n,k}) + o_p(\sqrt{n})\label{eq:ineq_ell}
	\end{align}
	As \(\left(\ell_{n,k}\right)_{n \geq 0}\) is a non-decreasing sequence of positive integers, and that \(\ell_{n,k} \leq n\), we can apply \eqref{eq:prelim-2} of Lemma \ref{lem:diff-U-approx} and we get
	\begin{align*}
		U_{n,k} - U_{\ell_{n,k},k} &= (n - \ell_{n,k})[-(1 - \alpha)v_k + o(1)] + M_{n,k} - M_{\ell_{n,k},k} + \ell_{n,k} (\hat{\rho}_{\ell_{n,k},k} - \hat{\rho}_{n,k})\\
		&\leq (n - \ell_{n,k})[-(1 - \alpha)v_k + o(1)] + |M_{n,k} - M_{\ell_{n,k},k}| + \ell_{n,k} |\hat{\rho}_{\ell_{n,k},k} - \hat{\rho}_{n,k}|\\
		&= (n - \ell_{n,k})[-(1 - \alpha)v_k + o_p(1)] + o_p(\sqrt{n})\\
		&\leq o_p(\sqrt{n}) 
	\end{align*}
	Which prove \eqref{eq:ineq-U-1}. For the second inequality, we proved in Theorem \ref{thm:LLN} that \(n(\hat{\rho}_n - v) = O\left(\sqrt{n \log \log n}\right) \quad \text{a.s.}\) which also gives that
	\begin{align*}
		\ell_{n,k}(\hat{\rho}_{\ell_{n,k}} - \hat{\rho}_{n,k}) &= O\left(\sqrt{n \log \log n}\right) + O\left(\sqrt{\ell_{n,k} \log \log \ell_{n,k}}\right) \\
		&= O\left(\sqrt{n \log \log n}\right) \quad \text{a.s.}
	\end{align*}
	Also, $\{M_{n,k}\}_{n \geq 1}$ is a martingale, so by LIL for martingales (see \cite{Stout1970}) we have \(M_{n,k} = O\left(\sqrt{n \log\!\log n}\right)\) then
	\begin{align*}
		M_{n,k} - M_{\ell_{n,k},k} &= O\left(\sqrt{n \log\!\log n}\right) + O\left(\sqrt{\ell_{n,k} \log\!\log \ell_{n,k}}\right)\\
		&= O\left(\sqrt{n \log\!\log n}\right) \quad \text{a.s.}
	\end{align*}
	So we obtain the second inequality \eqref{eq:ineq-U-2}
	\begin{align*}
		U_{n,k} - U_{\ell_{n,k},k} &\leq (n - \ell_{n,k})[-(1 - \alpha)v_k + o(1)] + |M_{n,k} - M_{\ell_{n,k},k}| + \ell_{n,k} |\hat{\rho}_{\ell_{n,k},k} - \hat{\rho}_{n,k}|\\			
		&\leq O\left(\sqrt{n \log \log n}\right) \quad a.s.
	\end{align*}
\end{proof}

Finally, we recall Lemma~\ref{lem:componentwise} and give its proof below.

\jointcltlemma*
\begin{proof}
	This proof is based on the central limit theorem for martingales.
	
	\textbf{Step 1. Defining the martingale.}
	
	For \(k \in [K]\) we define \( Y_{i,k} = \xi_{i,k} - \theta_k \) and \(M_{n,k} = \sum_{i=1}^{n} X_{i,k} Y_{i,k}.\) So \(\{M_{n,k}\}_{n \geq 0}\) is a martingale as we have
	\begin{align*}
		\mathbb{E}[\Delta M_{i,k} \mid \mathcal{F}_{i-1}] 
		&= \mathbb{E}[M_{i,k} - M_{i-1,k} \mid \mathcal{F}_{i-1}]\\
		&= \mathbb{E}[X_{i,k} Y_{i,k} \mid \mathcal{F}_{i-1}]\\
		&= X_{i,k} \, \mathbb{E}[Y_{i,k}]\\
		&= 0
	\end{align*}
	The quadratic variation of \( \{M_{n,k}\}_{n \geq 0} \) can be expressed as:
	\begin{align*}
		\langle M_{.,k} \rangle_n &= \sum_{i=1}^{n} \mathbb{E}\left[ (\Delta M_i)^2 \mid \mathcal{F}_{i-1} \right]\\
		&= \sum_{i=1}^{n} \mathbb{E}\left[ X_{i,k} Y_{i,k}^2 \mid \mathcal{F}_{i-1} \right]\\
		&= \sum_{i=1}^{n} X_{i,k}\mathbb{E}\left[Y_{i,k}^2 \mid \mathcal{F}_{i-1} \right]\\
		&= V_k N_{n,k}
	\end{align*}
	And for \(k \neq j \in [K]\) we have
	\[ \left<M_{.,k}, M_{.,j}\right>_n = \sum_{s=1}^{n} \mathbb{E}\left[\left(\xi_{s,k} - \theta_{k}\right)\left(\xi_{s,j} - \theta_{j}\right)\underbrace{X_{s,k}X_{s,j}}_{= 0}\mid \mathcal{F}_{s-1}\right] = 0\]	
	\textbf{Step 2. Lindeberg Condition for the martingale.}
	
	Let $\varepsilon > 0$, we need to prove that for all \(k \in [K]\)
	\begin{equation*}
		L_{n,k}(\varepsilon) 
		:= \frac{1}{\langle M_{.,k} \rangle_n} 
		\sum_{i=1}^n 
		\mathbb{E}\!\left[
		(\Delta M_{i,k})^2 
		\mathbf{1}_{\{|\Delta M_{i,k}| > \varepsilon \sqrt{\langle M_{.,k} \rangle_n}\}}
		\mid \mathcal{F}_{i-1}
		\right]
		\overset{\mathbb{P}}{\longrightarrow} 0 
		\quad \text{as } n \to +\infty.
	\end{equation*}
	By dominated convergence theorem and the fact that $\mathbb{E}\!\left[ |Y_{1,k}|^{2} \right] < +\infty$ we have that 
	\[\mathbb{E}\!\left[Y_{1,k}^2  \mathbf{1}_{\{|Y_{1,k}| > t\}}
	\right] \underset{t \to +\infty}{\longrightarrow} 0\]
	For \(\eta > 0\), choose \(m := m(\eta)\) such that 
	\[
	\frac{1}{V_k}
	\, \mathbb{E}\!\left[
	Y_{1,k}^2 
	\mathbf{1}_{\{|Y_{1,k}| > \varepsilon \sqrt{V_k} \sqrt{m}\}}
	\right] \leq \eta
	\]
	Now consider the event \(\left\{N_{n,k} \geq m\right\}\), we 
	have \(\mathbf{1}_{\{|Y_{1,k}| > \epsilon \sqrt{V_k} \sqrt{N_{n,k}}\}} \leq \mathbf{1}_{\{|Y_{1,k}| > \epsilon \sqrt{V_k} \sqrt{m}\}}\). So we have 	
	\begin{align*}
		L_{n,k}(\varepsilon)
		&\le 
		\frac{1}{V_k N_{n,k}} 
		\sum_{i=1}^n 
		X_{i,k} 
		\mathbb{E}\!\left[
		Y_{i,k}^2 
		\mathbf{1}_{\{|Y_{i,k}| > \varepsilon \sqrt{V_k} \sqrt{N_{n,k}}\}}
		\mid \mathcal{F}_{i-1}
		\right] \\
		&\leq 
		\frac{1}{V_k}
		\, \mathbb{E}\!\left[
		Y_{1,k}^2 
		\mathbf{1}_{\{|Y_{1,k}| > \varepsilon \sqrt{V_k} \sqrt{m}\}}
		\right] \\
		&\leq \eta
	\end{align*}
	So \(\mathbb{P}\left(L_{n,k}(\epsilon) \leq \eta\right) \leq \mathbb{P}\left(N_{n,k} \leq m\right) \underset{n \to +\infty}{\longrightarrow} 0\)	
	.Therefore, the Lindeberg condition is verified, and by the martingale CLT theorem, we have:
	\begin{equation}\label{eq:martingale-clt}
		\frac{M_n}{\sqrt{\langle M \rangle_n}} 
		\overset{\mathcal{D}}{\longrightarrow} 
		\mathcal{N}(0,I_K)
	\end{equation}
	
	And we have from the estimator definition \(D_n\left(\hat{\Theta}_n - \Theta\right) = \mathrm{diag}\left(\sqrt{V_1}, \ldots, \sqrt{V_K}\right)\frac{M_n}{\sqrt{\langle M \rangle_n}}\), then using \eqref{eq:martingale-clt} we get
	\begin{equation*}
		D_n\left(\hat{\Theta}_n - \Theta\right)
		\overset{\mathcal{D}}{\longrightarrow} 
		\mathcal{N}\left(0,\mathrm{diag}\left(V_1, \ldots, V_K\right)\right)
	\end{equation*}
\end{proof}

We are now ready to prove Theorem~\ref{thm:normality-gen-erade}.

\begin{proof}[Proof of Theorem \ref{thm:normality-gen-erade}]
	\begin{enumerate}[label=(\roman*)]
		\item Using Lemma \ref{lem:U-ineq} and inequalities \eqref{eq:prelim-4} and \eqref{eq:condition} from Lemma \ref{lem:preliminary} give that	
		\begin{equation}\label{eq:ineq-proof-4}
			\forall k \in [K]: \, N_{n,k} - n\hat{\rho}_{n,k} \leq o_p(\sqrt{n}) \quad \text{and} \quad N_{n,k} - n\hat{\rho}_{n,k} \leq O(\sqrt{n \log \log n}) \quad \text{a.s.}
		\end{equation}
		And we have that 
		\begin{align*}
			n\hat{\rho}_{n,k} - N_{n,k} &= n - N_{n,k} - n (1 - \hat{\rho}_{n,k})\\
			&= \sum_{j=1; j\neq k}^{K} \left(N_{n,j} - n \hat\rho_{n,j}\right)\\
			&\leq \sum_{j=1; j\neq k}^{K} o_p(\sqrt{n}) = o_p(\sqrt{n}) 
		\end{align*}
		Similarly for the second inequality we obtain that
		 \(\,n\hat{\rho}_{n,k} - N_{n,k} \leq O\left(\!\sqrt{n \log\!\log(n)}\right)\) \textit{a.s.}, so with \eqref{eq:ineq-proof-4} we obtain \eqref{eq:thm-res-1} and \eqref{eq:thm-res-2},
		\[
		|N_{n,k} - n\hat{\rho}_{n,k}| = o_p(\sqrt{n}) \quad \text{and} \quad  |N_{n,k} - n\hat{\rho}_{n,k}| = O(\sqrt{n \log\!\log n}) \quad \text{a.s.}
		\]
		now using Theorem \ref{thm:LLN} we have
		\begin{align*}
			N_{n,k} - n v_k &= N_{n,k} - n \hat{\rho}_{n,k} + n(\hat{\rho}_{n,k} - v_k)\\
			&= O(\sqrt{n \log \log n}) \quad a.s.
		\end{align*}
		which proves \eqref{eq:thm-res-3}.

		\item Using the fact that \(\sqrt n D_n^{-1}\) converges to \(\mathrm{diag}\left(\frac{1}{\sqrt{v_1}},\ldots,\frac{1}{\sqrt{v_K}}\right)\), the CLT result of Lemma \ref{lem:componentwise} combined with Slutsky's lemma yields the first asymptotic normality result:
		\begin{align*}
			\sqrt n(\hat\Theta_n-\Theta) &= \sqrt n D_n^{-1} D_n(\hat\Theta_n-\Theta)\\
			&\xrightarrow{\mathcal{D}}\mathrm{diag}\left(\frac{1}{\sqrt{v_1}},\ldots,\frac{1}{\sqrt{v_K}}\right) \mathcal N\!\big(0,\mathrm{diag}\left(V_1,\dots,V_K\right)\big)\\
			&= \mathcal N(0,V) 
		\end{align*}
		
		And by using \eqref{eq:approx-rho-Op} we obtain that \(\sqrt{n} (\hat{\rho}_n - v) = \sqrt{n} (\hat{\Theta}_n - \Theta) G + O_p\!\left(\frac{1}{\sqrt{n}}\right)\), now from \(O_p\!\left(\frac{1}{\sqrt{n}}\right) \xrightarrow{\mathbb{P}} 0\) and the CLT result above on \(\sqrt n(\hat\Theta_n-\Theta)\), we can apply Slutsky theorem giving
		\begin{equation}\label{eq:normality-1}
			\sqrt{n} (\hat{\rho}_n - v) \xrightarrow{\mathcal{D}} \mathcal{N}(0, GVG^{\top})
		\end{equation}
		And we have
		\begin{align*}
			\sqrt{n} \left( \frac{N_n}{n} - v \right) 
			&= \sqrt{n} \left( \frac{N_n}{n} - \hat{\rho}_n \right)
			+ \sqrt{n} (\hat{\rho}_n - \nu)\\
			&= \sqrt{n} (\hat{\rho}_n - \nu) + o_p(1)
		\end{align*}
		Again with \eqref{eq:normality-1} we can apply the Slutsky lemma:
		\[
		\sqrt{n} \left( \frac{N_n}{n} - v \right) \xrightarrow{\mathcal{D}} \mathcal{N}(0,GVG^{\top})
		\]
		Noting \(W_n := \sqrt{n} \left( \frac{N_n}{n} - \nu \right) 
		\quad \text{and} \quad Y_n := \sqrt{n} (\hat{\rho}_n - \nu)\), then
		\[
		\begin{pmatrix}
			W_n \\
			Y_n
		\end{pmatrix}
		=
		\begin{pmatrix}
			Y_n \\
			Y_n
		\end{pmatrix}
		+
		\underbrace{\begin{pmatrix}
				o_p(1) \\
				0
		\end{pmatrix}}_{\xrightarrow{\mathbb{P}} 0}
		\]
		And as we have
		\[
		\begin{pmatrix}
			Y_n \\
			Y_n
		\end{pmatrix}
		\xrightarrow{\mathcal{D}} 
		\mathcal{N}\left(0,\begin{pmatrix}
			GVG^{\top} & GVG^{\top} \\
			GVG^{\top} & GVG^{\top}
		\end{pmatrix}
		\right),
		\]
		we can use the Slutsky lemma again and get
		\[
		\begin{pmatrix}
			W_n \\
			Y_n
		\end{pmatrix}
		\xrightarrow{\mathcal{D}} \mathcal{N}(0,\Omega).
		\]
	\end{enumerate}	
\end{proof}

%
%

		\section{Supporting Proofs for $\alpha$RTS-FE}\label{appnB}
		
\subsection{Proof of Theorem \ref{thm:forced}}
In this part we redefine the last hitting time $\ell_{n,k}$ for \(\alpha\)RTS-FE designs family as follow
\begin{equation}\label{eq:stopping-time-def-fe}
	\ell_{n,k} := \max \left\{\, m \leq n : \frac{N_{m,k}}{m} \leq \hat{\rho}_{m,k} \text{ or } \left(k \in arg\!\min_{j \in U_m} N_{m,j}\text{ and } \mathbf{U}_m \neq \emptyset\right) \,\right\}
\end{equation}

We prove below that the conditions in Lemma~\ref{lem:generic-conditions-main} are valid for any design in the \(\alpha\)RTS-FE family and for the last hitting time defined as in \eqref{eq:stopping-time-def-fe}, which preserves the results of all Theorems for their counterparts without forced exploration.

\begin{lemma}\label{lem:preliminary2}
	For any \(k \in [K]\) we have the following results
	\begin{enumerate}[label=(\roman*)]
		\item \(\left(\ell_{n,k}\right)_{n \geq 1}\) is a non-decreasing sequence and \(\forall n \in \mathbb{N}: \,\ell_{n,k} \leq n\)
		\item for all $k\in [K]$, $n \in \mathbb{N}$,\begin{equation}\label{eq:prelim-4}
			N_{n,k} - n\hat{\rho}_{n,k} \leq 1+N_{\ell_{n,k},k} - \ell_{n,k} \hat{\rho}_{\ell_{n,k},k} + U_{n,k} - U_{\ell_{n,k},k}
		\end{equation}
		\item for all $k\in [K]$, $n \in \mathbb{N}$,
		\begin{equation}\label{eq:condition}
			N_{\ell_{n,k},k} - \ell_{n,k} \hat{\rho}_{\ell_{n,k},k} \leq o(\sqrt{n}), \ a.s.
			\quad
		\end{equation}
	\end{enumerate}
\end{lemma}

\begin{proof}[Proof of Lemma \ref{lem:preliminary2}]
	Consider \(\ell_{n,k}\) as defined in \eqref{eq:stopping-time-def-fe}, we have
	\begin{enumerate}[label=(\roman*)]
		\item Obvious from definition.
		\item Let \(k \in [K]\), if \(m \geq \ell_{n,k} + 1\) then
		\[
		\frac{N_{m,k}}{m} > \hat{\rho}_{m,k} \text{ and } \left\{k \notin arg\!\min_{j \in U_m} N_{m,j}\text{ or } \mathbf{U}_m = \emptyset\right\}
		\]
		If \(\mathbf{U}_m \neq \emptyset\) then \(k \notin arg\!\min_{j \in U_m} N_{m,j}\), so \(p_{m+1,k} = 0\) following \eqref{eq:fe-phase} . Otherwise if \(\mathbf{U}_m = \emptyset\) then \(p_{m+1,k} \leq \alpha \hat{\rho}_{m,k}\) following \eqref{eq:fe-phase}, so the result follows using the same arguments as in the previous case.
		\item At \(m = \ell_{n,k}\), if \(\frac{N_{\ell_{n,k},k}}{\ell_{n,k}} \leq \hat{\rho}_{\ell_{n,k},k}\) then \(	N_{\ell_{n,k},k} - \ell_{n,k} \hat{\rho}_{\ell_{n,k},k} \leq 0
		\),
		otherwise if \(\mathbf{U}_{\ell_{n,k}} \neq \emptyset\) and \(k \in arg\!\min_{j \in \mathbf{U}_{\ell_{n,k}}} N_{\ell_{n,k},j}\), then for \(a \in \mathbf{U}_{\ell_{n,k}}\) we have
		\begin{align*}
			N_{\ell_{n,k},k} - \ell_{n,k} \hat{\rho}_{\ell_{n,k},k}
			&\leq N_{\ell_{n,k},k}\\
			&\leq N_{\ell_{n,k},a}\\
			&\leq h(\ell_{n,k})
		\end{align*}
		So either way \(N_{\ell_{n,k},k} - \ell_{n,k} \hat{\rho}_{\ell_{n,k},k} \leq o(\sqrt{n}) \quad \text{a.s.}\)
	\end{enumerate}
\end{proof}

\subsection{$\alpha$RTS-FE with Sparse Allocations}\label{proofs:sparse-case}

\begin{proof}[Proof of Theorem \ref{thm:sparse-rate}]
	We first prove that \(N_{n,k} \to +\infty\). Let us assume towards contradiction that
	\[
	\exists i \in [K],\ \exists M \in \mathbb{N},\ \exists N \in \mathbb{N}: \, \forall n \geq N:\ N_{n,i} = M
	\]
	As \(h(n) \to +\infty\) we can select \(N\) such that \(\forall n \geq N: \, h(n) \geq M \). Now consider \( n \geq N \), we have
	\begin{align*}
		N_{n,i} \leq h(n) 
		& \Rightarrow \ i \in U_n \\
		& \Rightarrow U_n \neq \emptyset \\
		& \Rightarrow I_{n+1} \in \arg\min_{j \in U_n} N_j(n) \\
		& \Rightarrow N_{n, I_{n+1}} \leq N_{n,i} = M
	\end{align*}
	Since \( I_{n+1} \) has been selected at round \( n+1 \), we have \(N_{n+1, I_{n+1}} = N_{n, I_{n+1}} + 1 \leq M + 1\), and for \( j \neq I_{n+1} \) we have \(N_{n+1,j} = N_{n,j}\). 
	
	Now define \(D(n) := \sum_{j=1}^K \left(M + 1 - N_j(n)\right)^+\), then we have
	\begin{align*}
		D(n+1) - D(n)
		&= \sum_{j=1}^K \left(M + 1 - N_{n+1, j}\right)^+ 
		- \sum_{j=1}^K \left(M + 1 - N_{n,j}\right)^+ \\
		&= \left(M + 1 - N_{n+1, I_{n+1}}\right)^+ 
		- \left(M + 1 - N_{n, I_{n+1}}\right)^+ \\
		&= \left(M - N_{n, I_{n+1}}\right)^+ 
		- \left(M - N_{n, I_{n+1}} + 1\right)^+ \\
		&= -1.
	\end{align*}
	Thus \(D\left(n + D(n) + 1\right) = -1\) which is a contradiction as \(D \geq 0\). Therefore,
	\[
	\forall i \in [K], \quad N_{n,i} \xrightarrow[n\to\infty]{} +\infty.
	\]
	From the result above we have that \(\hat{\theta}_{n,k}-\theta_k \to 0\) using the law of large numbers. Additionally, by using the Law of Iterated Logarithm for martingales we obtain that
	\[\hat{\Theta}_n - \Theta = O\left(\sqrt{\frac{\log\log(N_{n,k})}{N_{n,k}}}\right) \quad a.s.\]
	In the same spirit of the proof of Theorem \ref{thm:LLN}, we use the continuity of \(\rho\) and Lemma \ref{lem:convergence} to obtain that \(\lim\limits_{n \to +\infty} \frac{N_{n,k}}{n} = \lim\limits_{n \to +\infty} \hat{\rho}_{n,k} = v_k\)
\end{proof}

\subsection{Asymptotic Distribution of the Pearson Test Statistic}
\begin{proof}[Proof of Lemma \ref{lem:pearson}]
	Let \(Y_n := \left( \sqrt{N_{n,k}}(\hat{\theta}_{n,k} - \theta)\right)_{k \in [K]}\) and \(s_n:=\left(\sqrt{\omega_{n,k}}\right)_{k \in [K]}\) where \(\omega_{n,j}:=\frac{N_{n,j}}{n}\), then we rewrite \(\mathcal{X}_n = \frac{1}{\tilde{\theta}_n (1 - \tilde{\theta}_n)} 
	\sum_{k=1}^K N_{n,k} (\hat{\theta}_{n,k} - \tilde{\theta}_n)^2\)
	and we have
	\begin{align*}
		\sqrt{N_{n,k}} (\tilde{\theta}_n - \theta)
		&= \sqrt{N_{n,k}} \sum_{j=1}^K \omega_{n,j} (\hat{\theta}_{n,j} - \theta)\\
		&= \sqrt{\frac{N_{n,k}}{n}} \sum_{j=1}^K \sqrt{\omega_{n,j}}\sqrt{N_{n,j}} (\hat{\theta}_{n,j} - \theta)\\
		&= \sqrt{\omega_{n,k}} \sum_{j=1}^K \sqrt{\omega_{n,j}} \, [Y_n]_j
	\end{align*}
	so 
	\begin{align*}
		\sqrt{N_{n,k}} (\hat{\theta}_{n,k} - \tilde{\theta}_n) &= \sqrt{N_{n,k}} (\hat{\theta}_{n,k} - \theta) 
		- \sqrt{N_{n,k}} (\tilde{\theta}_n - \theta)\\
		&= [Y_n - s_n s_n^\top Y_n]_k
	\end{align*}
	then \(\mathcal{X}_n = \frac{1}{\tilde{\theta}_n (1 - \tilde{\theta}_n)} 
	Y_n^\top (I - s_n s_n^\top) Y_n\)
	, and we have from Lemma \ref{lem:componentwise} that \(\frac{1}{\sqrt{\theta(1-\theta)}} Y_n \xrightarrow{\mathcal{D}} \mathcal{N}(0, I_K)\) and from Theorem \ref{thm:LLN} we have \(\tilde{\theta}_n \to \theta\) and \(I - s_n s_n^\top \to I - s s^\top
	\) with \(s := (\sqrt{\nu_1}, \dots, \sqrt{\nu_K})\). So as $I - s s^\top$ is an orthogonal projection matrix with rank $K-1$, we obtain by Slutsky argument that \(\mathcal{X}_n \xrightarrow{\mathcal{D}} \chi^2_{K-1}\)
\end{proof}

		\section{Technical Lemmas}\label{appnC}
		\subsection{Probabilistic Technical Lemmas}
\begin{lemma}\label{lem:mart-1}
	Suppose $\{M_n, \mathcal{F}_n; n \geq 1\}$ is a martingale with $\mathbb{E}|\Delta M_n|^2 \leq C_0$, where $\Delta M_n = M_n - M_{n-1}$ and $\{\mathcal{F}_n\}$ is a filter of sigma fields with $\mathcal{F}_1 \subset \mathcal{F}_2 \subset \cdots$. Then, for any $L > 4$,
	\begin{align}
		\mathbb{E} \left[ \max_{m \leq L} |M_n - M_{n - m}| \right] 
		&\leq 4\sqrt{C_0 L}\label{eq:A1}\\
		\mathbb{E} \left[ \max_{m : L \leq m \leq n} \frac{|M_n - M_{n - m}|}{m} \right] 
		&\leq 12 \sqrt{\frac{C_0}{L}} \label{eq:A2}
	\end{align}
\end{lemma}
\begin{proof}
	\begin{enumerate}[label=(\roman*)]
		We first prove the first inequality, we have 
		\begin{align}
			\sup_{1 \le m \le L} |M_n - M_{n-m}| 
			&\le \sup_{1 \le m \le L} |M_n - M_{n-L}| + |M_{n-m} - M_{n-L}|\\
			&= |M_n - M_{n-L}| + \sup_{1 \le m \le L} |M_{n-m} - M_{n-L}|\notag\\
			&\le 2 \sup_{k \in \{n-L+1, \dots, n\}} |M_k - M_{n-L}|\notag\\
			&= 2 \sup_{k \in \{1, \dots, L\}} |S_k|\label{eq:doob-2}
		\end{align}
		where \(S_k := M_{n-L+k} - M_{n-L}.\) 
		
		Define now the new forward filtration \(\mathcal{G}_k := \mathcal{F}_{n-L+k}\). 
		We have 
		\begin{align*}
			\mathbb{E}[S_k \mid \mathcal{G}_{k-1}] 
			&= S_{k-1} + \mathbb{E}[\Delta S_k\mid \mathcal{G}_{k-1}]\\
			&= S_{k-1} + \mathbb{E}[M_{n-L+k} - M_{n-L+k-1} \mid \mathcal{G}_{k-1}]\\
			&= S_{k-1}
		\end{align*}
		so \(\{S_k\}\) is a \(\mathcal{G}_k\)-martingale. Then \(\{S_k^2\}\) is a sub-martingale. Applying Doob's maximal inequality (see Corollary 2.1.6 in \cite{marc-yor}), we obtain
		\begin{align}
			\mathbb{P}\!\left( \sup_{k \in \{1,\dots,L\}} S_k^2 \ge x^2 \right)
			&\le \frac{\mathbb{E}[S_L^2]}{x^2}\notag\\
			\implies \mathbb{P}\!\left( \sup_{k \in \{1,\dots,L\}} |S_k| \ge x \right)
			&\le \frac{\mathbb{E}[S_L^2]}{x^2}\label{eq:doob-1}
		\end{align}
		\(\{S_k\}\) being a martingale gives that \(\{\Delta S_k\}\) are orthogonal,  
		i.e. \(\operatorname{Cov}(\Delta S_k, \Delta S_j)=0\). So:
		\begin{align*}
			\mathbb{E}[S_L^2]
			&= \sum_{k=1}^L \mathbb{E}[\,|\Delta S_k|^2\,]\\
			&= \sum_{k=1}^L \mathbb{E}\!\left[\,|\Delta M_{n-L+k}|^2\,\right]\le L C_0.		
		\end{align*}	
		So from \eqref{eq:doob-2} and \eqref{eq:doob-1} we obtain 
		\[
		\mathbb{P}\!\left( \sup_{m \in \{1,\dots,L\}} |M_n - M_{n-m}| \ge x \right)
		\le 4\frac{LC_0}{x^2}
		\]	
		Then
		\begin{align*}
			\mathbb{E}\left[\sup_{m \in \{1,\dots,L\}} |M_n - M_{n-m}|\right] &= \int_{0}^{+\infty}\mathbb{P}\!\left( \sup_{m \in \{1,\dots,L\}} |M_n - M_{n-m}| \ge x \right) dx\\
			&\leq \int_{0}^{2\sqrt{LC_0}}1 \dot dx + 4\int_{2\sqrt{LC_0}}^{+\infty}\frac{LC_0}{x^2} \dot dx = 4\sqrt{LC_0}
		\end{align*}
		
		\item \begin{align*}
			\mathbb{E} \left[ \max_{L \leq m \leq n} \frac{|M_n - M_{n - m}|}{m} \right]
			&\leq \sum_{\log_2(L) \leq j \leq \log_2(n)} 
			\frac{1}{2^{j - 1}} \sqrt{2} \, \mathbb{E} \left[ 
			\max_{2^{j - 1} \leq m < 2^j} |M_n - M_{n - m}| \right] \\
			&\leq \sum_{j \geq \log_2(L)} 
			\frac{1}{2^{j - 1}} \, \mathbb{E} \left[ 
			\max_{m < 2^j} |M_n - M_{n - m}| \right] \\
			&\leq \sum_{j \geq \log_2(L)} 
			\frac{4\sqrt{C_0 \cdot 2^j}}{2^{j - 1}} \\
			&= 4\sum_{j \geq \log_2(L)} \frac{\sqrt{C_0}}{\sqrt{2^{j - 2}}}
			\leq 12 \sqrt{\frac{C_0}{L}}	
		\end{align*}
	\end{enumerate}
	
\end{proof}

\begin{lemma}\label{lem:expectation-to-O}
	Let $(u_n)_{n\ge1}$ be a sequence of positive numbers and let $\{X_n\}_{n\ge1}$ be real-valued random variables such that there exists $C>0$ for which \(\forall n \geq 1, \ \mathbb{E}|X_n|\le C u_n\). 
	Then 
	\[X_n = O_p(u_n).\]
\end{lemma}

\begin{proof}
	By Markov's inequality applied to the non-negative random variable $|X_n|$,
	\[
	\mathbb{P}\big(|X_n|> M u_n\big)
	\;\le\; \frac{\mathbb{E}|X_n|}{M u_n}
	\;\le\; \frac{C}{M}.
	\]
	Given $\varepsilon>0$, choose $M\ge C/\varepsilon$. Then, for all $n$,
	\[
	\mathbb{P}\big(|X_n|> M u_n\big)\le \varepsilon.
	\]
	This is precisely the definition of $X_n=O_p(u_n)$.
\end{proof}

\begin{lemma}\label{lem:chebychev-mart}
	Suppose  \(\{X_n\}_{n \geq 1}\) is a square-integrable martingale with \(\mathbb{E}[X_1]=0\). Then for \(\lambda > 0\):
	\[
	\mathbb{P}\left(\sup_{1\leq i \leq n} X_i \geq \lambda\right)\leq \frac{\mathbb{V}ar(X_n)}{\mathbb{V}ar(X_n)+ \lambda^2}
	\]
\end{lemma}
\begin{proof}
	For \(\alpha > 0\) we have
	\[
	\mathbb{P}\left(\sup_{0 \leq k \leq n} X_k \geq c \right) = \mathbb{P}\left(\sup_{0 \leq k \leq n} (X_k + \alpha)^2 \geq (c+\alpha)^2 \right)
	\]
	As \(\{X_n\}\) is a martingale then \(\{X_n^2\}\) is a sub-martingale, so we can apply maximal Doob's inequality (see Theorem 12.4.2 in \cite{jf-le-gall}), which gives
	
	\begin{align*}
		\mathbb{P}\left(\sup_{0 \leq k \leq n} (X_k + \alpha)^2 \geq (c + \alpha) \right) &\leq \frac{\mathbb{E}\left[(X_n + c)^2\right]}{(c + \alpha)^2}\\
		&= \frac{\mathbb{V}ar\left[X_n\right] + \alpha^2}{(c + \alpha)^2}\\
		&= \frac{\mathbb{V}ar\left[X_n\right] + \left(\frac{\mathbb{V}ar\left[X_n\right]}{c}\right)^2}{\left(c + \frac{\mathbb{V}ar\left[X_n\right]}{c}\right)^2} \text{; Taking \(\alpha = \frac{\mathbb{V}ar\left[X_n\right]}{c}\)}\\
		&= \frac{\mathbb{V}ar\left[X_n\right]}{\mathbb{V}ar\left[X_n\right] + c^2}
	\end{align*}
	which proves the claim.		
\end{proof}

\begin{proposition}\label{prop:exp-fam}
	For regular one-parameter exponential families with the parameter of interest $\theta_k$ being the mean parameter, we have $\mathbb{V}ar(\xi_{1,k}) = I_k(\theta_k)^{-1}$.
\end{proposition}
\begin{proof}
	Consider the response of treatment $k$ following a natural exponential family
	\[
	f(y \mid \eta_k)
	=
	\exp\{\eta_k y - A(\eta_k)\} h(y),
	\]
	so that \(\mathbb{E}_{\eta_k}[Y] = A'(\eta_k) \text{ and }
	\mathbb{V}ar_{\eta_k}(Y) = A''(\eta_k).\) If the parameter of interest is the mean \(\theta_k = A'(\eta_k),\)
	and we take $\xi_{1,k}=Y$, then
	\[
	\mathbb{E}[\xi_{1,k}] = \theta_k,
	\qquad
	\mathbb{V}ar(\xi_{1,k}) = A''(\eta_k).
	\]
	Moreover, by the Fisher information reparameterization formula,
	\[
	I_k(\theta_k)
	=
	I_{\eta_k}(\eta_k)
	\left(\frac{d\eta_k}{d\theta_k}\right)^2
	=
	A''(\eta_k)\left(\frac{1}{A''(\eta_k)}\right)^2
	=
	\frac{1}{A''(\eta_k)},
	\]
	which yields \(\mathbb{V}ar(\xi_{1,k}) = I_k(\theta_k)^{-1}.\)
	This includes, for example, Bernoulli, Poisson, exponential, and normal models with known variance.
\end{proof}

\subsection{Analytical Technical Lemmas}
\begin{lemma}\label{lem:convergence}
	Let $(a_n)_{n \geq 1}$ a positive increasing sequence such that $\forall n \in \mathbb{N}: a_n \leq n$ and $\alpha, v \in [0,1]$, $(X_n)_{n \geq 1}$ a real sequence such that $\lim\limits_{n \to +\infty} \frac{X_n}{n} = -\alpha v$, then for any \(\epsilon_n \leq o(1)\) we have
\begin{align*}
		\lim\limits_{n \to +\infty}\left( \epsilon_n + \frac{X_n - X_{a_n}}{n} \right)^+ &= 0 
\end{align*}
\end{lemma}

\begin{proof}	
	\textbf{Case 1:} $(a_n)_{n \in \mathbb{N}}$ is bounded
	
	As \((a_n)\) is increasing sequence then \(a_n\) is stationary so \(\frac{X_{a_n}}{n} \longrightarrow 0\) , which gives
	\begin{align*}
		\lim_{n \to +\infty} \epsilon_n + \frac{X_n - X_{a_n}}{n} = -\alpha v
		&\implies \lim_{n \to +\infty} \left(\epsilon_n + \frac{X_n - X_{a_n}}{n}\right)^+ = 0
	\end{align*}	
	\textbf{Case 2:} $a_n \to +\infty$
	
	In this case we have \(\lim\limits_{n \to +\infty} \frac{X_{a_n}}{a_n} = -\alpha v\), then:
	\begin{align*}
		\forall \varepsilon > 0,\ \exists N \in \mathbb{N},\ \forall n \geq N &,\quad -\varepsilon \leq \frac{X_{a_n}}{a_n} + \alpha v \leq \varepsilon \quad \text{a.s.}\\		
		&\Rightarrow -\varepsilon - \alpha v \leq \frac{X_{a_n}}{a_n} \leq \varepsilon - \alpha v \quad \text{a.s.}\\
		&\Rightarrow -\frac{X_{a_n}}{n} \leq \varepsilon \frac{a_n}{n} + \alpha v \frac{a_n}{n} \leq \varepsilon + \alpha v \quad \text{a.s.}\\
		&\Rightarrow \epsilon_n + \frac{X_n - X_{a_n}}{n} \leq \epsilon_n + \frac{X_n}{n} + \varepsilon + \alpha v \quad \text{a.s.}		
	\end{align*}	
	So as \(\lim\limits_{n \to +\infty} \frac{X_n}{n} + \alpha v = 0\,\) and \(\lim\limits_{n \to +\infty} \epsilon_n \leq 0\), we obtain that \(\lim\limits_{n \to +\infty} \epsilon_n + \frac{X_n - X_{a_n}}{n} \leq 0 \,\), then 
	\[\lim\limits_{n \to +\infty} \left(\epsilon_n +\frac{X_n - X_{a_n}}{n}\right)^+ = 0 \]
\end{proof}

\begin{lemma}\label{lem:mart-2}
	For all \( \ell \geq 1, L < n \) and any sequence \( (Q_n)_{n \geq 0} \in \mathbb{R}^d \), we have
	\[
	\|Q_n - Q_\ell\| \leq (n - \ell) \max_{L \leq m \leq n} \frac{\|Q_n - Q_{n - m}\|}{m} + \max_{m < L} \|Q_n - Q_{n - m}\|.
	\]
\end{lemma}

\begin{proof}
	Let \( d = n - \ell \), we consider two cases. If \( d < L \) then \(\|Q_n - Q_\ell\| = \|Q_n - Q_{n - d}\| \leq \max_{m < L} \|Q_n - Q_{n - m}\|.\)
	Otherwise if \( d \geq L \) then
	\begin{align*}
		\|Q_n - Q_\ell\| &= d.\,\frac{\|Q_n - Q_{n-d}\|}{d} \\		
		& \leq \left(n - \ell\right) \max_{L \leq m \leq n} \frac{\|Q_n - Q_{n-m}\|}{m}
	\end{align*}
	So gathering the two upper bounds for both cases we get:	
	\begin{align*}
		\|Q_n - Q_\ell\| &\leq (n - \ell) \max_{L \leq m \leq n} \frac{\|Q_n - Q_{n - m}\|}{m} + \max_{m < L} \|Q_n - Q_{n - m}\|,
	\end{align*}	
\end{proof}

	\end{appendix}

\end{document}